\documentclass[final,11pt]{amsart}
\usepackage[T1]{fontenc}
\usepackage{a4wide,mathtools,bbm,amsthm,amsmath,amsfonts}
\usepackage{amssymb,color,epic,multirow,comment}
\usepackage[noend]{algpseudocode}
\usepackage[ruled]{algorithm2e}
\usepackage{graphicx}
\usepackage{subcaption}
\usepackage[export]{adjustbox}
\usepackage{tikz,pgfplots}
\usepackage{orcidlink}
\usepackage{float}
\usepackage{booktabs}
\usepackage{tikz-cd, comment}
\usepgfplotslibrary{groupplots}
\pgfplotsset{compat=newest}
\usetikzlibrary{arrows,decorations,backgrounds,positioning,arrows.meta, decorations.markings}
\usepgfplotslibrary{colorbrewer}

\mathtoolsset{showonlyrefs}
\theoremstyle{plain}

\newtheorem{theorem}{Theorem}[section]

\newtheorem{proposition}[theorem]{Proposition}

\newtheorem{definition}[theorem]{Definition}
\newtheorem{remark}[theorem]{Remark}

\def\Letters{A,B,C,D,E,F,G,H,I,J,K,L,M,N,O,P,Q,R,S,T,U,V,W,X,Y,Z}
\makeatletter
\@for \@l:=\Letters \do{%
  \expandafter\edef\csname\@l bb\endcsname{%
  \noexpand\ensuremath{\noexpand\mathbb{\@l}}}%
  \expandafter\edef\csname\@l bf\endcsname{{\noexpand\bf \@l}}%
  \expandafter\edef\csname\@l cal\endcsname{%
  \noexpand\ensuremath{\noexpand\mathcal{\@l}}}%
  \expandafter\edef\csname\@l eu\endcsname{%
  \noexpand\ensuremath{\noexpand\EuScript{\@l}}}%
  \expandafter\edef\csname\@l frak\endcsname{%
  \noexpand\ensuremath{\noexpand\mathfrak{\@l}}}%
  \expandafter\edef\csname\@l rm\endcsname{{\noexpand\rm \@l}}%
  \expandafter\edef\csname\@l scr\endcsname{%
  \noexpand\ensuremath{\noexpand\mathscr{\@l}}}%
}
\newcommand{\bs}[1]{{\boldsymbol#1}}

\newcommand{\isdef}{\mathrel{\mathrel{\mathop:}=}}
\newcommand{\defis}{\mathrel{=\mathrel{\mathop:}}}

\definecolor{navy}{RGB}{102,153,255}
\definecolor{tuerkis}{RGB}{51,153,204}
\algrenewcommand\alglinenumber[1]{\ding{\numexpr191 + #1}}

\title[Bespoke multiresolution analysis of graph signals]
{Bespoke multiresolution analysis of graph signals}
\author{Giacomo Elefante\orcidlink{0000-0001-5576-6802}, Gianluca Giacchi\orcidlink{0000-0002-6809-1311}, Michael Multerer\orcidlink{0000-0003-0170-0239}, and Jacopo Quizi\orcidlink{0009-0001-9199-2812}}
\address{
IDSIA USI-SUPSI,
Universit{\`a} della Svizzera italiana,
Via la Santa 1, 6962 Lugano, Svizzera.}
\email{\{giacomo.elefante,gianluca.giacchi,michael.multerer,jacopo.quizi\}@usi.ch}
\begin{document}
\begin{abstract}
We present a novel framework for discrete multiresolution analysis of graph
signals. The main analytical tool is the samplet transform, originally defined
in the Euclidean framework as a discrete wavelet-like construction, tailored to
the analysis of scattered data. The first contribution of this work is defining
samplets on graphs. To this end, we subdivide the graph into a fixed number of
patches, embed each patch into a Euclidean space, where we construct samplets,
and eventually pull the construction back to the graph.
This ensures orthogonality, locality, and the vanishing moments property with
respect to properly defined polynomial spaces on graphs. Compared to classical
Haar wavelets, this framework broadens the class of graph signals that can
efficiently be compressed and analyzed. Along this line, we provide a
definition of a class of signals that can be compressed using our construction.
We support our findings with different examples of signals defined on graphs
whose vertices lie on smooth manifolds. For
efficient numerical implementation, we combine heavy edge clustering, to
partition the graph into meaningful patches, with landmark \texttt{Isomap},
which provides low-dimensional embeddings for each patch. Our results
demonstrate the method's robustness, scalability, and ability to yield
sparse representations with controllable approximation error, significantly
outperforming traditional Haar wavelet approaches in terms of compression
efficiency and multiresolution fidelity.
\end{abstract}

\maketitle

\section{Introduction}\label{sec:introduction}
Due to their effectiveness in structuring data, encoding and visualizing data,
graphs play a fundamental role in modern signal analysis, cp.\ \cite{SNF13}. 
To mention a few, they find applications in social networks, see
\cite{IMG19-sociology}, medical signal processing, see
\cite{GWH22-biology,HOD21-biology}, neuroscience, see \cite{BS23,DMC24},
computer engineering \cite{LYL25,Xu2023}, acoustics, see \cite{BRZ23},
and transportation, see \cite{MAMDJ14,NOH18,WXZ19}. As a consequence, the
analysis of graph signals is a highly relevant task.
Multiresolution analysis has always been a core tool in signal analysis because
of its capability of analyzing signals at different levels of resolutions,
providing valuable insights in their space and frequency structure.
The main tool of multiresolution analysis is the {\em wavelet transform}, see,
e.g., \cite{Mallat89} and the references therein. 

In the last decades, several
authors have been working on adaptations of the wavelet transform on graphs,
focusing on data embedded in high dimensional spaces, see \cite{CM10,LNW08}.
{\em Diffusion wavelets} rely on diffusion operators to generate orthonormal
scaling functions and wavelets for data compression and denoising on graphs,
cf. \cite{HM_CM06}. Similarly, {\em Spectral graph wavelets} introduced in
\cite{HVG11} are constructed by using the eigenfunctions of the graph
Laplacian. Further constructions are based on partitioning trees, exemplified
by the Haar-like wavelets on hierarchical trees proposed in \cite{AW24,gavish}.
In \cite{REC11}, this idea is extended by introducing data-adaptive orthonormal
wavelets on trees. A related extension for graph-based data is the construction
of {\em wedgelets} in \cite{Erb23}, which employ adaptive greedy partitioning
of graphs.
Recently, {\em samplets} emerged as a novel approach for multiresolution
analysis of scattered data, mimicking the role of wavelets in this unstructured framework, cf. \cite{HM22}. Samplets are localized, discrete signed measures
exhibiting vanishing moments and can, therefore, be constructed on general data
sets. So far, samplet theory has been developed in a Euclidean framework, with applications to compressed sensing \cite{BHM24}, deepfake recognition \cite{HVM25}, kernel learning \cite{AMW24,HMSS24} and local singularities detection \cite{avesani2025multiresolution}. 

In this article, we extend the samplet construction to graphs. To this end, we
first partition the graph into several patches, each of which we assume to
correspond to the discretization of some Riemannian manifold. Subsequently,
we embed each patch into a Euclidean space, where samplets are constructed,
and finally pulled back onto the corresponding patches of the graph.
This way, we result in a {\em samplet forest} for the analysis of graph
signals. The patch-wise approach is crucial. In comparison with the classical 
Haar-wavelet construction on graphs, which is invariant under rotations of the
parametric domain, samplets exhibit higher order vanishing moments. The
corresponding polynomials are not invariant under rotations and therefore
require the coordinate system to be fixed in advance. The result of this
construction is a wavelet-like basis tailored to the underlying graph, which
exhibits much stronger compression capabilities compared to Haar-wavelets.

In particular, we define classes of graph signals that can efficiently be
compressed using samplets. These classes describe graph signals that can
locally be approximated by {\em generalized polynomials}. The classes
resemble Jaffard's microlocal spaces, see \cite{jaffard1991pointwise},
within the discrete graph framework. We provide a proof of the compressibility
of such signals represented in a samplet basis.
For the efficient numerical implementation, we combine heavy edge clustering, 
see \cite{KK97,KK98} to partition the graph into meaningful patches, with
landmark \texttt{Isomap}, see \cite{MS02, ST02,Torgeson52},
which provides low-dimensional embeddings for each patch.

We support our findings with different instances of graph signal analysis.
Concretely, we consider signals defined on nearest neighbor graphs of a
unit square embedded into high dimension, of the Swissroll manifold
and of the Stanford bunny.

The paper is organized as follows. We first give, in Section~\ref{sec:graph}, a
definition of the graph framework we will be working with and define the
microlocal spaces of compressible functions. In Section~\ref{sec:samplets},
we locally define the samplets in the coordinate domains, which are then pulled
back to assemble a samplet forest. In Section~\ref{sec:sampletgraphs},
we prove that signals in microlocal spaces have decaying samplet coefficients,
where the rate of decay is limited by the number of vanishing moments.
Moreover, we discuss two different strategies for signal compression. 
For the approximation of coordinate maps, we consider multidimensional
scaling, and provide corresponding consistency error bounds in
Section~\ref{sec:embedding}.
Then, in Section~\ref{sec:numerics}, we report numerical results to support
the efficiency of the method proposed. Finally, we draw a conclusion in 
Section~\ref{sec:conclusion}.

\section{Graphs and microlocal spaces}\label{sec:graph}
In this section, we provide basic definitions and define smoothness classes
of graphs by transferring the concept of Jaffard's microlocal spaces to the
graph framework.
\subsection{Setting}\label{ssec:Setting}
Let \(G=(V,E,w)\) denote a weighted graph with vertices
\(V=\{v_1,\ldots,v_N\}\), edges \( E \subseteq \{(v_i, v_j)
\in V \times V : i\neq j\}\) and weight function
\(w\colon E \to[0,\infty)\). The weight function 
represents pairwise distances between vertices. In this regard, the graph
can be considered as a discretization of some (unknown) topological 
manifold \(\Mcal\). 
We assume that \(G\) can be decomposed into several
subgraphs \(G_1,\ldots,G_p\). These subgraphs correspond to a subdivision of
\(\Mcal\) into smooth patches \(\Ucal_1,\ldots,\Ucal_p\) such that
each of these patches is a Riemannian manifold of dimension $q$
with a unique chart \((\Ucal_r,\phi_r)\), for \(r=1,\ldots,p\),
where \(\phi_r\colon\mathcal{U}_r \to \phi_r(\mathcal{U}_r)\subseteq\mathbb{R}^q\) is the coordinate map.
In this sense, the edge
weights of $G$ can be interpreted as approximations of the
geodesic distances between the corresponding points on each patch 
of \(\Mcal\).
\subsection{Microlocal spaces on graphs}\label{subsectionEF}
With a slight abuse of notation, we shall identify
$\mathcal{U}_r = \{{v}_1, \ldots, {v}_{N_r}\}$.
The composition of a monomial with the coordinate map $\phi_r$ naturally
induces a definition of monomials on $\mathcal{U}_r$. More precisely, given a
monomial ${\bs x}^{\bs\alpha}$ on $\mathbb{R}^q$, we define the corresponding
function on $\mathcal{U}_r$ as
\begin{equation}
{\bs X}^{\bs\alpha}\isdef {\bs x}^{\bs\alpha}\circ \phi_r,
\end{equation}
see Figure \ref{Fig:Monomials} for a visual reference.
\begin{figure}[htb]
\centering
\resizebox{0.65\textwidth}{!}{
\begin{tikzpicture}[every node/.style={font=\small}]
 
\shade[ball color=blue!30, opacity=0.6] 
  (1,1.2) .. controls (2,2.8) and (4.5,2.2) .. (5.5,1.5) 
  .. controls (5.7,0.5) and (4,-0.5) .. (2,-0.2) 
  .. controls (1.2,0.2) and (1,0.6) .. (1,1.2);
 \node at (1,1.8) {\(G\)};
 
\draw[dashed, rounded corners=10pt, line width=0.8pt] (2.2,0.7) .. controls (3,1.5) and (3.5,1.2) .. (3.6,0.7) .. controls (3.5,0.2) and (2.5,0.2) .. (2.2,0.7);
\node at (2,1) {\(\mathcal{U}_r\)};
 
\foreach \x/\y in {2.7/0.7, 3/0.9, 3.3/0.8, 2.9/0.6, 3.1/0.5}
{
    \fill (\x,\y) circle (0.04);
}
 
\foreach \x/\y in {2.2/-2, 3.8/-1.9, 2.3/-3, 3/-2.7, 3.6/-2.5} {
    \fill (\x,\y) circle (0.04);
}
 
\draw[->] (2.9,0.1) -- (2.9,-1.3) node[midway, left] {\(\phi_r\)};
 
\draw[dashed, line width=0.8pt] (2,-3.5) rectangle (4,-1.5);
\node at (1,-3.5) {\(\phi_r(\mathcal{U}_r)=\mathbb{R}^q\)};
 
\draw[->, bend left=30] (4,1.3) to node[above right] {\({\bs X}^\alpha\)} (8,1.3);
\draw[->, bend right=30] (4.3,-2.5) to node[below right] {\({\bs x}^\alpha\)} (8.3,1);
 
\node at (8.3,1.3) {\(\mathbb{R}\)};
 
\end{tikzpicture}
}
\caption{A schematic representation of the definition of the monomials ${\bs X}^\alpha$.}
\label{Fig:Monomials}
\end{figure}
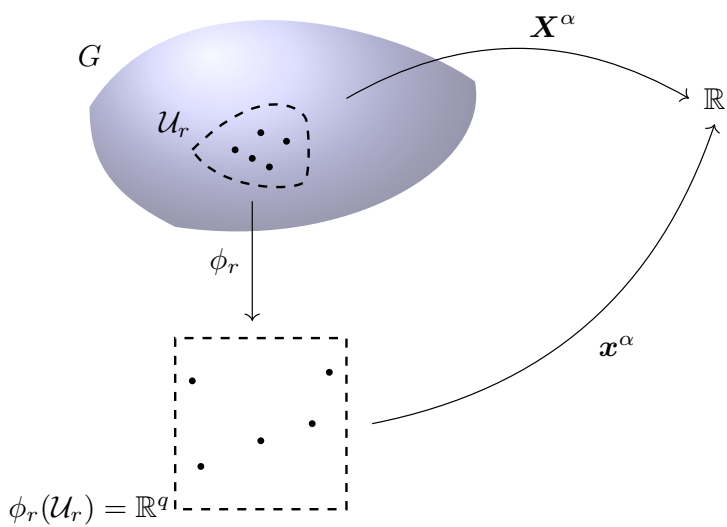

Now given ${v}_i \in \Ucal_r$ we introduce the evaluation functional
\begin{equation}
    \langle \delta_{{v}_i}, {{\bs X}}^{\bs\alpha} \rangle\isdef \phi_r({v}_i)^{\bs\alpha},
\end{equation}

Next, we use the manifold structure constructed on the graph to define spaces of
\emph{locally regular} functions on the graph. We mimic the definition of Jaffard's microlocal spaces \cite{jaffard1991pointwise}, where the emphasis is on the local approximation of a signal through polynomials of a fixed degree. 
Let $\Omega\subseteq\mathbb{R}^q$ be a domain, $\bs x_0\in\Omega$ and $\gamma\geq0$. A function $f\colon\Omega\to \mathbb{R}$ is in the class $C^\gamma(\bs x_0)$ if there exists $R>0$ and a polynomial $P$ of degree $\lfloor\gamma\rfloor$, so that
\begin{equation}\label{EuclJaff}
     |f(\bs x)-P(\bs x -\bs x_0)|\leq C \|\bs x-\bs x_0\|^\gamma
\end{equation}
holds for every $\bs x\in B_R(\bs x_0)$ and a constant \(C>0\). Here, \(\|\cdot\|\) denotes the Euclidean norm.
We want to discuss a similar definition for a graph signal $f\colon G\to \mathbb{R}$ in our
framework. One may be tempted to define $f\in C^\gamma_G(v_0)$, where the lower-script $G$ denotes that $f$ is defined on a graph $G$, and $v_0\in\Ucal_r$ is a fixed vertex of $G$, if $f\circ\phi_r^{-1}\in C^\gamma\big(\phi_r(v_0)\big)$. However, $f$ is only defined on $G$ and therefore $f\circ \phi_r^{-1}$ is not defined on any ball around $\phi_r(v_0)$. Moreover, translations are not naturally defined on $G$, and this complicates the replacement of $P(\bs x-\bs x_0)$ in \eqref{EuclJaff}. 
Even so, in order to define sparsity classes for graph signals, we require a notion of
smoothness classes playing the same role as Jaffard's spaces in the Euclidean framework. Precisely, the microlocal spaces $C^\gamma$ are used to obtain decay estimates for the \emph{samplet coefficients} of functions defined on domains of $\mathbb{R}^d$, see \cite[Theorem 3.1]{avesani2025multiresolution}. We define spaces of functions on $G$ to obtain analogous decay estimates for graph signals as follows.

\begin{definition}\label{defCalpha}
   Let $G=(V,E,w)$ be a graph and $v_0\in V$. Let $\gamma\geq0$, $C>0$, and $(\mathcal{U}_r,\phi_r)$ 
   be the unique chart containing $v_0$. We write $f\in C^\gamma_G(C,v_0)$ if there exists a sequence of real coefficients $(c_{\bs\beta})_{|\bs\beta|\leq\lfloor\gamma\rfloor}$, $\sum_{|\bs\beta|=\lfloor\gamma\rfloor}|c_{\bs\beta}|\neq0$, such that
   \begin{equation}\label{DefCaG}
       \bigg| f(v) -\!\!\sum_{|\bs \beta|\leq\lfloor\gamma\rfloor}\!\! c_{\bs\beta} \big(\phi_r(v)-\phi_r(v_{0})\big)^{\bs\beta}\bigg|\leq C d_G(v,v_0)^\gamma
   \end{equation}
   for every $v\in\mathcal{U}_r$, where $d_G$ denotes the graph distance.
\end{definition}

We remark that Definition~\ref{defCalpha} is a more general and localized version of the 
\((C,\gamma)\)-H\"older classes for \(0<\gamma<1\) introduced in \cite{gavish}
to general \(\gamma\geq 0\).
\section{Construction of samplets on graphs}\label{sec:samplets}
Different from the original construction of samplets in \cite{HM22}, which
employs a single tree structure for the construction of the multiresolution
hierarchy, we consider here an ensemble of trees resulting in a samplet
forest. More precisely, given a decomposition of the underlying graph into
several patches, we construct a samplet tree for each of these patches.
The idea is to construct multiresolution hierarchies and samplets 
on the co-domains \(\phi_r(\Ucal_r)\) and then pull the construction back 
to the graph using \(\phi_r^{-1}\).

The first step is to introduce a hierarchical structure on the co-domain
\(\phi_r(\Ucal_r)\) of each patch. This hierarchy may then be pulled back
to \(G\) using \(\phi_r^{-1}\).
We remark that it is well known that the resulting hierarchy on \(G\) naturally induces a multiresolution analysis with an associated Haar wavelet basis. 
We refer the reader to \cite{DD97,LNW08,M07}. A prototypical example of 
this construction can be found in \cite{gavish} and we remark that lowest
order samplets resemble the Haar wavelets on trees considered there.
We base our construction of the multilevel hierarchy on the concept of a \emph{cluster tree} for each patch \(\Ucal_r\), \(r=1,\ldots,p\).

\begin{definition}\label{def:cluster-tree}
Let \(X\subset\Rbb^q\) and let 
$\mathcal{T}=(V,E)$ be a tree with vertices $V$ 
and edges $E$.
We define its set of leaves as
\(
\mathcal{L}(\mathcal{T})\isdef\{\tau\in V\colon\tau~\text{has no children}\}.
\)
The tree $\mathcal{T}$ is a \emph{cluster tree} for
\(X\), if
      \(X\) is the {root} of $\mathcal{T}$ and
all $\tau\in V\setminus\mathcal{L}(\mathcal{T})$
are disjoint unions of their children.
The \emph{level} \(j_\tau\) of $\tau\in\mathcal{T}$ is its distance from
the root and the bounding box \( B_\tau \) is the smallest axis-parallel cuboid that contains all points of \( \tau \). The depth of the cluster tree is given by \( J \isdef \max_{\tau \in \Tcal} j_\tau\).
\end{definition}

For the construction of the binary tree we adopt geometric clustering by successively subdividing the bounding box of the embedded patch along the longest axis at the midpoint.

Next, we associate a cluster tree \(\Tcal_r\) to each co-domain $\phi(\Ucal_r)$, 
\(r=1,\ldots,p\) and introduce a \emph{two-scale} transform between basis 
elements associated to a cluster $\tau\in\Tcal_r$ of level
$j$ and its child clusters on level $j+1$. 
To this end, we define \emph{scaling distributions}
$\mathbf{\Phi}_{j}^{\tau} = \{ \varphi_{j,i}^{\tau} \}_i$ and
\emph{samplets} $\mathbf{\Psi}_{j}^{\tau} = \{\psi_{j,i}^{\tau} \}_i$
as linear combinations of the scaling 
distributions $\mathbf{\Phi}_{j+1}^{\tau}$ of $\tau$'s child clusters. 
By denoting the number of elements by \(n_{j+1}^\tau
\isdef|\mathbf{\Phi}_{j+1}^{\tau}|\), this results in the
\emph{refinement relations}
\begin{equation}
\varphi_{j,i}^{\tau}
=\sum_{\ell=1}^{n_{j+1}^\tau}q_{j,\Phi,\ell,i}^{\tau}\varphi_{j+1,\ell}^{\tau}
\end{equation}
and
\begin{equation}
\psi_{j,i}^{\tau}
=\sum_{\ell=1}^{n_{j+1}^\tau}q_{j,\Psi,\ell,i}^{\tau}
\varphi_{j+1,\ell}^{\tau},
\end{equation}
for certain coefficients \(q_{j,\Phi,\ell,i}^{\tau}
=[{\bs Q}_{j,\Phi}^{\tau}]_{\ell,i}\) and
\(q_{j,\Psi,\ell,i}^{\tau}=[{\bs Q}_{j,\Psi}^{\tau}]_{\ell,i} \).
These relations
may be written in matrix notation as
\begin{equation}\label{eq:refinementRelation}
\begin{bmatrix}\mathbf{\Phi}_{j}^{\tau}, \mathbf{\Psi}_{j}^{\tau}
\end{bmatrix}
 \isdef 
 \mathbf{\Phi}_{j+1}^{\tau}{\bs Q}_j^{\tau}
=
 \mathbf{\Phi}_{j+1}^{\tau}\begin{bmatrix} {\bs Q}_{j,\Phi}^{\tau},{\bs Q}_{j,\Psi}^{\tau}
\end{bmatrix}.
\end{equation}

Based on Section \ref{subsectionEF}, we define the moment matrix 
${\bs M}_j^\tau \in \mathbb{R}^{m_s} \times n_j^\tau$ as
the matrix with entries
\begin{equation}\label{eq:mom_mat_new}
    {\bs M}_j^\tau\isdef [\langle \delta_{{v}_i},{\bs X}^{\bs\alpha}\rangle]_{\bs\alpha, i} = \left[\phi({v}_i)^{\bs\alpha}\right]_{\bs\alpha, i},
\end{equation}
where \(m_s={s+q\choose q}\) is the dimension of the space of all
polynomials defined on \(\Rbb^q\) with degree \(|\bs\alpha|\leq s\).
Similar to the construction of samplets 
in the Euclidean space, see \cite{HM22},
we employ the QR decomposition of the moment matrix to
obtain filter coefficients for samplets with vanishing moments of order \(s+1\). By letting
\begin{equation}\label{eq:QR} 
  ({\bs M}_{j+1}^{\tau})^\intercal  = {\bs Q}_j^\tau{\bs R}
  \defis\big[{\bs Q}_{j,\Phi}^{\tau} ,
  {\bs Q}_{j,\Psi}^{\tau}\big]{\bs R},
 \end{equation}
the moment matrix 
for the cluster's own scaling distributions and samplets is given by 
\begin{equation}\label{eq:vanishingMomentsQR}
  \big[{\bs M}_{j,\Phi}^{\tau}, {\bs M}_{j,\Psi}^{\tau}\big]
= {\bs M}_{j+1}^{\tau} [{\bs Q}_{j,\Phi}^{\tau} , {\bs Q}_{j,\Psi}^{\tau} ]
  = {\bs R}^\intercal.
\end{equation}
Since ${\bs R}^\intercal$ is a lower triangular matrix, the first $k-1$ 
entries in its $k$-th column are zero. This corresponds to 
$k-1$ vanishing moments for the $k$-th distribution generated 
by the transformation
${\bs Q}_{j}^{\tau}=[{\bs Q}_{j,\Phi}^{\tau} , {\bs Q}_{j,\Psi}^{\tau} ]$. 
By defining the first $m_{s}$ distributions as scaling distributions and 
the remaining ones as samplets, we obtain samplets with vanishing moments of
order \(s+1\), i.e.,
\[
\langle\psi_{j,i}^\tau,{\bs X}^{\bs\alpha}\rangle=0\quad
\text{for }|\bs\alpha|\leq s.
\]

For leaf clusters of \(\Tcal_r\), we define the scaling distributions
by the Dirac measures supported at the points ${\bs\zeta_i} \in \phi_r(\Ucal_r) $, i.e., $\bs\Phi^\tau_J\isdef \{\delta_{\bs \zeta_i} : {\bs \zeta}_i=\phi_r(v_i), {v}_i \in \Ucal_r\}$, where $\phi_r(v_i)$ is defined as in Subsection~\ref{ssec:Setting}. Then collecting the samplets of all levels
together with the scaling distributions of the root cluster yields a 
samplet basis for \(\phi_r(\Ucal_r)\) and, by pulling it back via 
\(\phi_r^{-1}\) for \(\Ucal_r\), respectively. Each such samplet basis
gives rise to an orthogonal transformation matrix 
\({\bs T}_r\in\Rbb^{N_r\times N_r}\). Collecting these matrices in the
block-diagonal matrix \({\bs T}=\operatorname{diag}({\bs T}_1,\ldots,{\bs T}_p)\in\Rbb^{N\times N}\) gives rise to the orthogonal transformation matrix
for \(G\). We remark that the samplet transform \({\bs T}\) is usually not
constructed explicitly but rather applied by recursion, resulting in a
computational cost of \(\Ocal(N)\), see \cite{HM22} for details.
A visualization of samplets with four vanishing moments constructed on a patch 
of an \(\varepsilon\)-nearest neighbors graph of the Stanford bunny is shown
in Figure~\ref{fig:BunnySamplets}.

\begin{figure}[htb]
\centering
\begin{tikzpicture}
  \draw(0.5,0) node{\includegraphics[scale = 0.04,clip,trim = 700 200 750 350]{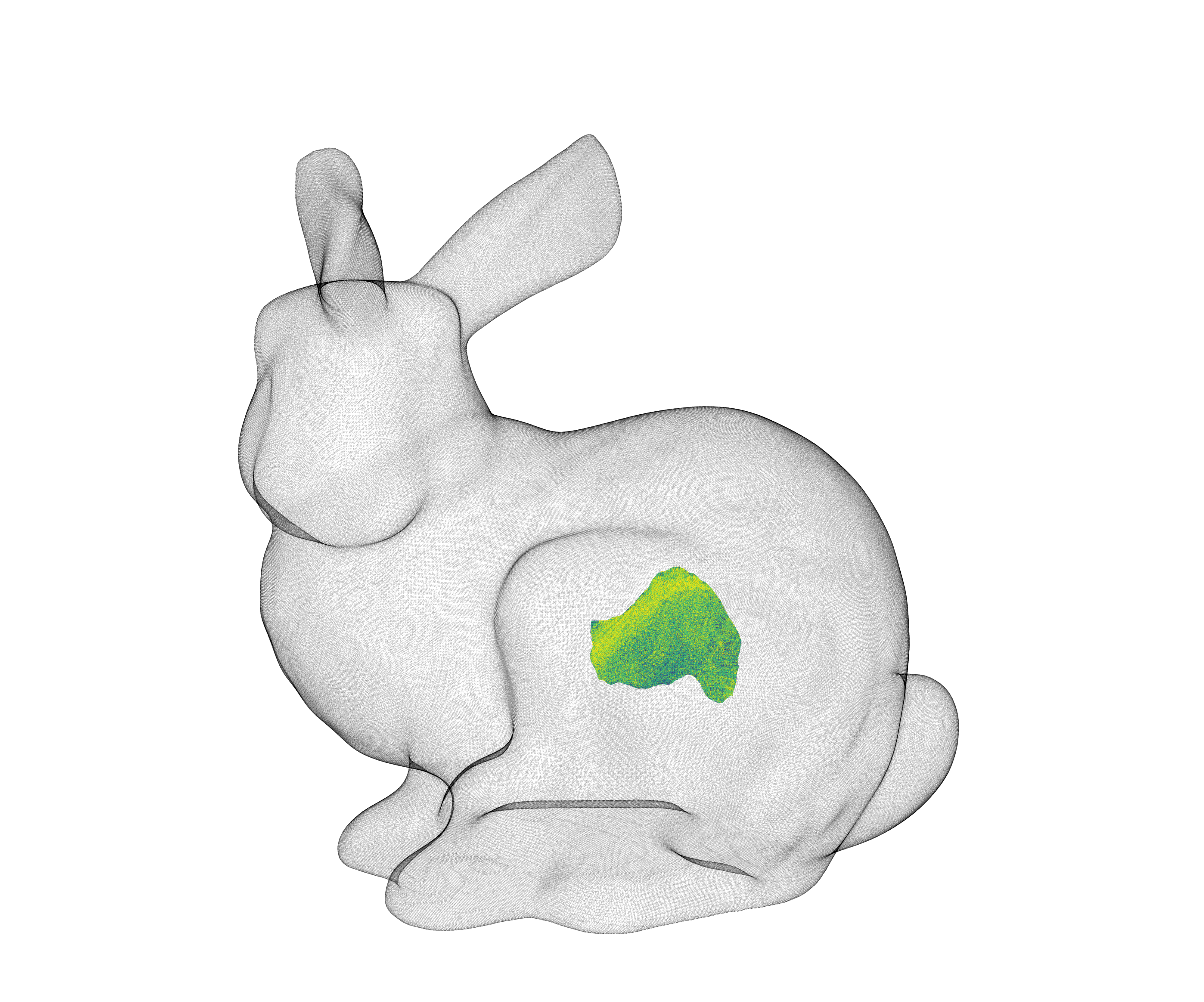}};
  \draw(-1.5,-3) node{\includegraphics[scale = 0.02,trim = 700 200 400 350, clip]{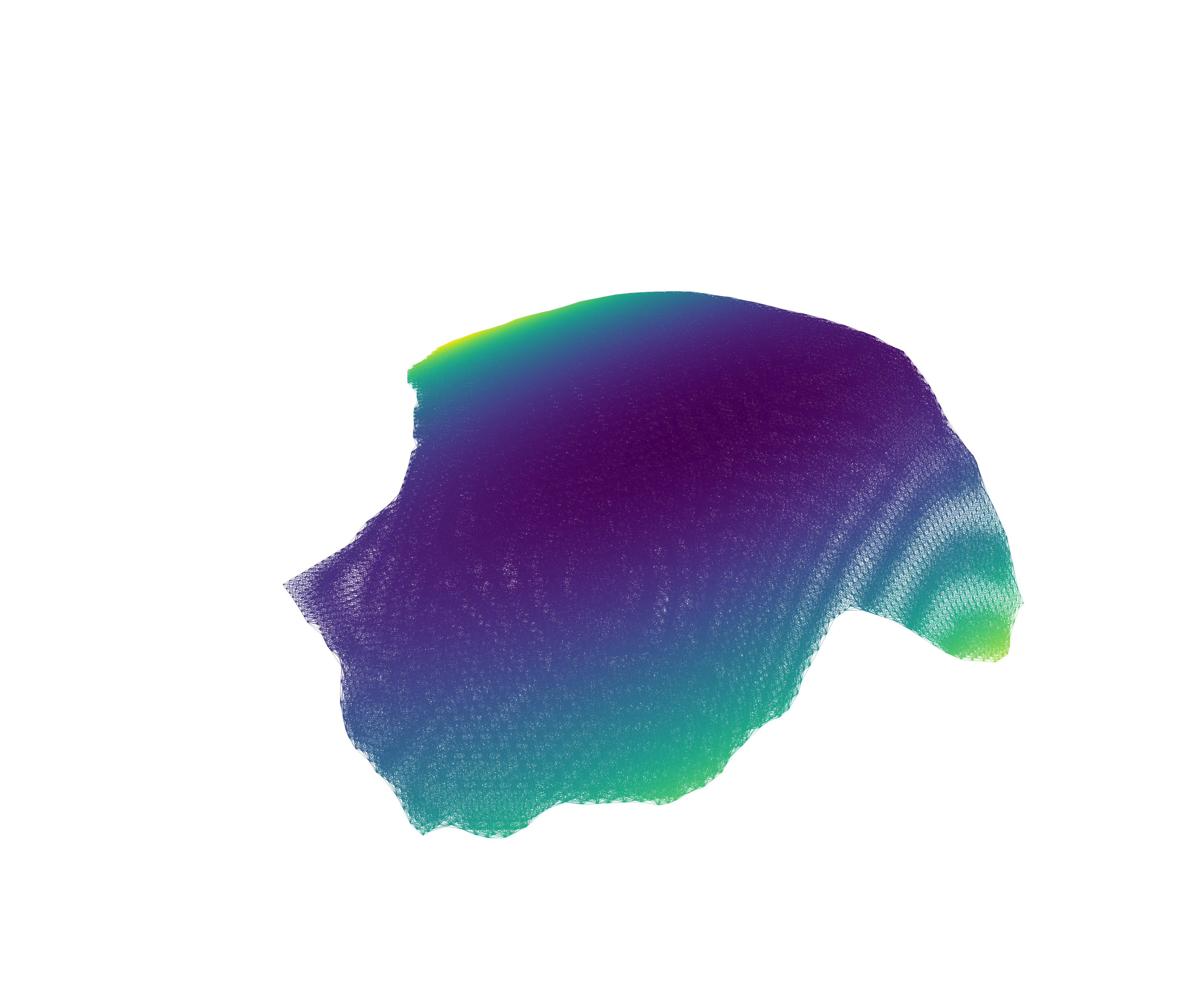}};
    \draw(2.5,-3) node{\includegraphics[scale = 0.02,trim = 700 200 400 350, clip]{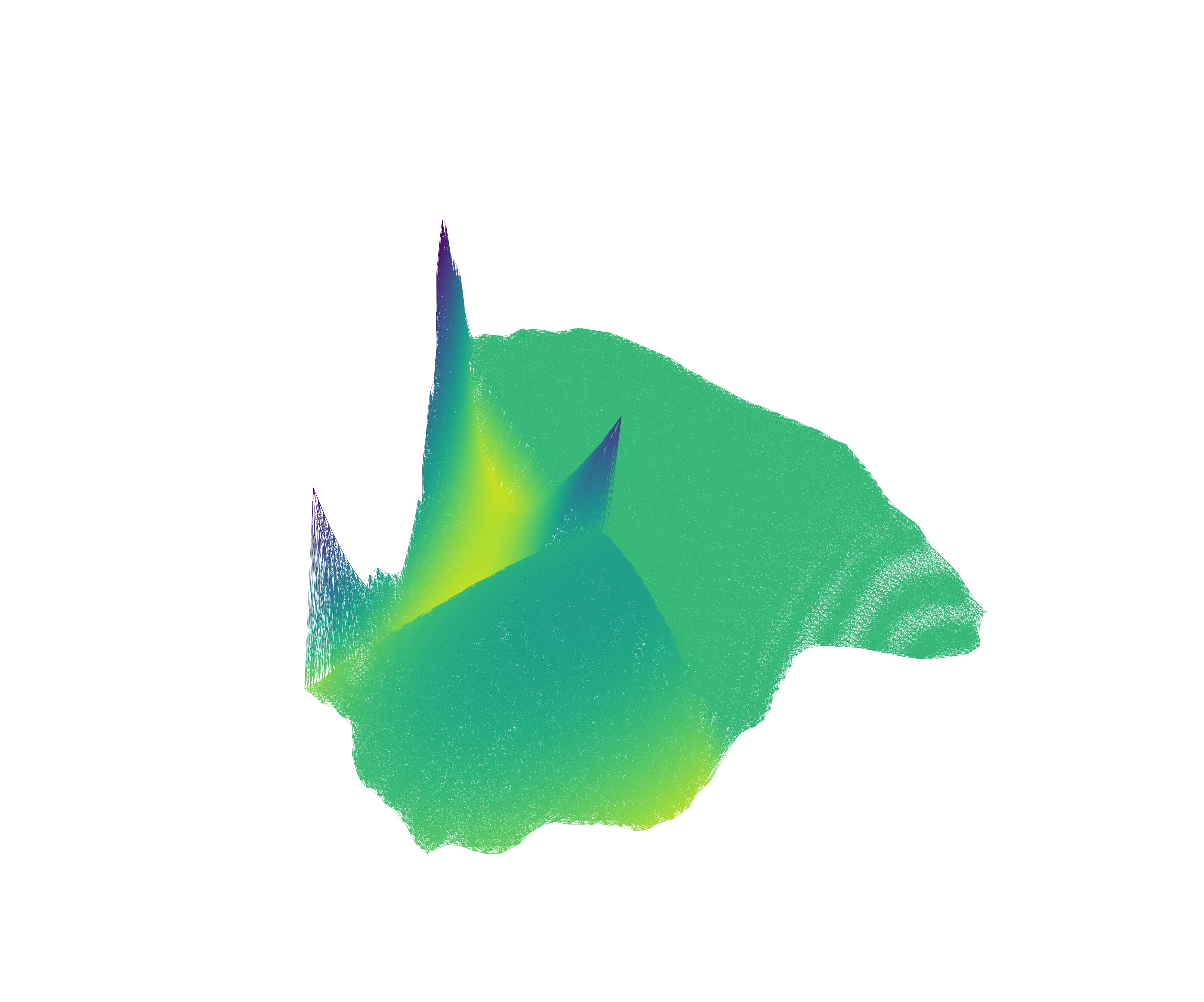}};
    \draw(-1.5,-5) node{\includegraphics[scale = 0.02,clip,trim = 700 200 400 350]{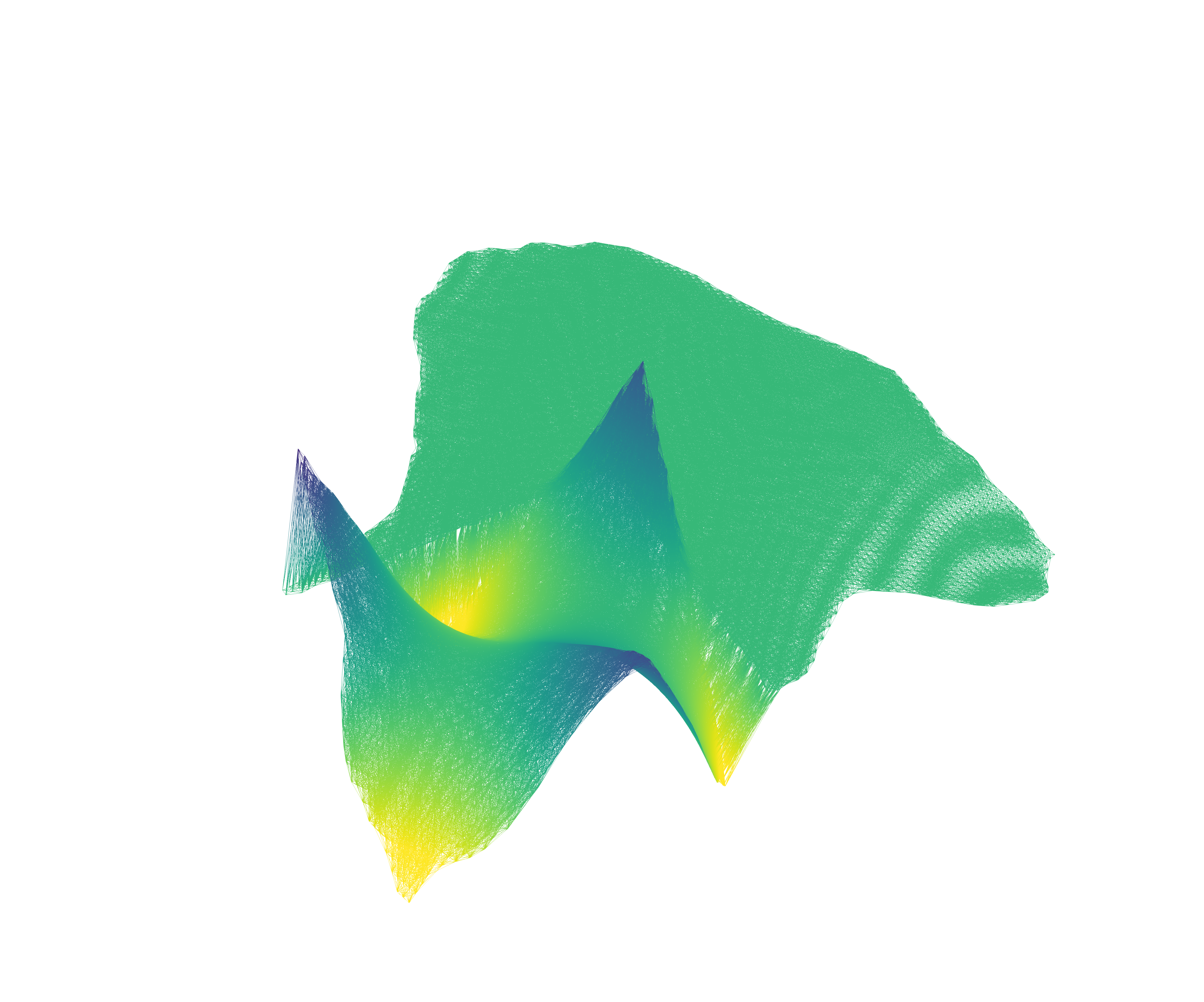}};
    \draw(2.5,-5) node{\includegraphics[scale = 0.02 ,clip,trim = 100 100 100 100]{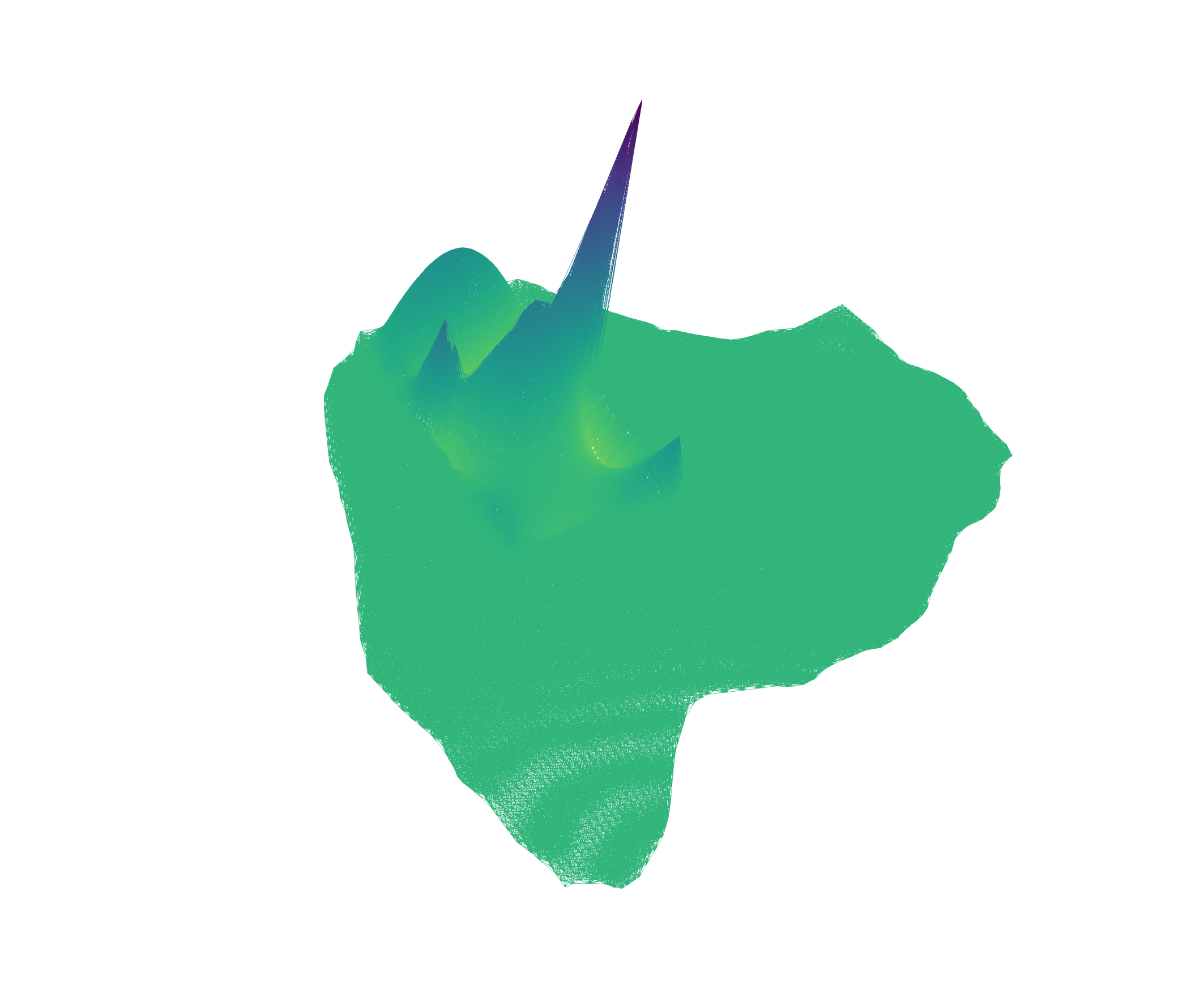}};
\end{tikzpicture}

\caption{\label{fig:BunnySamplets}The top row shows the Stanford bunny and a selected patch. The second row shows one scaling distribution (left) and a samplet on level 0 (right), and the last row shows samplets on level 0 and 1.}
\end{figure}

Using the aforementioned construction with a forest of cluster trees 
defined on each co-domain \(\phi_r(\Ucal_r)\), \(r=1,\ldots,p\), we end up
with \(p\) disjoint samplet bases.
This way, we avoid orientation
ambiguities, as each coordinate system is fixed, and samplets constructed
on \(\phi_{r}(\Ucal_r)\) respect its local geometry. Furthermore, we remark
that, if all cluster trees are balanced in the sense that \(J\sim\log(N)\)
and \(|\tau|\sim 2^{J-j_\tau}\), then the samplet basis can be constructed with
linear cost \(\Ocal(N)\).

\section{Samplet graph signal analysis}\label{sec:sampletgraphs}
In this section, we prove a decay result for the samplet coefficients of signals in classes $C^{\gamma}_G(C,v_0)$, thus justifying the compressibility of such functions in microlocal spaces. Then, we propose a compression strategy based on adaptive tree coarsening.
\subsection{Decay estimates of samplet coefficients}
Now, we establish decay estimates for the samplet coefficients of functions on graphs in our framework. We use here the spaces $C^\gamma_G$ of Definition \ref{defCalpha}. In the following, we assume that samplets have been constructed on $G$ so that the vanishing moments property holds, as described in Section \ref{sec:samplets}, up to order $\lfloor\gamma\rfloor$.
\begin{proposition}\label{propGG}
    Let $f\in C^\gamma_G(C,v_0)$, $\gamma\geq 0$, $C>0$. Let $(\mathcal{U},\phi)$ be the unique chart containing $v_0$. 
    Then, for every cluster $\tau$ that contains $v_0$, we have
    \begin{equation}\label{eqgg1graph}
        |\langle \psi_{j,k},f\circ\phi^{-1}\rangle|\leq C \max_{ v_j\in\mathcal{U}}d( v_j, v_0)^\gamma \sqrt{|\tau|}.
    \end{equation}
\end{proposition}
\begin{proof}
   Let us write $\psi_{j,k}=\sum_{\ell=1}^{|\tau|}\omega_{j,k}^{(\ell)}\delta_{\bs\zeta_\ell}$, where $\bs\zeta_\ell=\phi( v_\ell)$ and we are using the notation of Section \ref{sec:samplets}. By the vanishing moments property,
    \begin{align*}
\langle\psi_{j,k},f\circ\phi^{-1}\rangle
        &=\bigg\langle\psi_{j,k},f\circ\phi^{-1}-\sum_{|\bs\beta|\leq\lfloor\gamma\rfloor}c_{\bs\beta}\big(\cdot-\phi(v_0)\big)^{\bs\beta}\bigg\rangle\\
        &=\sum_{\ell=1}^{|\tau|} \omega_{j,k}^{(\ell)}\bigg( f(v_\ell) 
        - \sum_{|\bs\beta|\leq\lfloor\gamma\rfloor}c_{\bs\beta}\big(\phi( v_\ell)-\phi(v_0)\big)^{\bs\beta}\bigg).
    \end{align*}
    By Cauchy-Schwartz, using that $\sum_{\ell=1}^{|\tau|}(\omega_{j,k}^{(\ell)})^2=1$ and \eqref{DefCaG}, we obtain \eqref{eqgg1graph}.
\end{proof}

\subsection{Compression strategy} \label{subsec:sparse}
To compress a given graph signal \(f\colon V\to\Rbb\),
we partition the signal according to the given patches
and consider \(f|_{\Ucal_r}\), \(r=1,\ldots,p\).
Next, we apply the \emph{adaptive tree coarsening}
from \cite{HM_BD97} to each restriction of the signal to the patches, i.e., \(f|_{\Ucal_r}\),
represented in samplets coordinates. More precisely,
fixing a patch, we set 
\({\bs f}_r^\Delta\isdef[f(v)]_{v\in\Ucal_r}\) and
compute \({\bs f}_r^\Sigma={\bs T}_r{\bs f}_r^\Delta\).
Next, following the procedure in \cite{HM_BD97},
we construct a subtree of \(\Scal_r\subset\Tcal_r\) with \emph{energy}
\(e(\Scal_r)\geq(1-\varepsilon^2)\big\|{\bs f}_r^\Sigma\big\|^2\)
for some \(\varepsilon>0\). Herein, the energy of a node
is defined by
\begin{equation}\label{HM_eg:energy}
e(\tau)\isdef\big\|{\bs f}^\Sigma_r|_\tau\big\|^2 +\sum_{\tau'\in\operatorname{child}(\tau)}e(\tau')
,
\end{equation}
Defining the energy this way, \(e(\tau)\) is the contribution of the subtree with root \(\tau\) to the squared norm 
\(\big\|{\bs f}^\Sigma_r\big\|^2\). In particular, we have 
\(e(\Tcal_r)=\big\|{\bs f}^{\Sigma}_r\big\|^2\).

Based on the energies \eqref{HM_eg:energy}, we next define
\begin{equation}
\tilde{e}(\tau')\isdef q(\tau)\isdef
\frac{\sum_{\mu\in\operatorname{child}(\tau)}e(\mu)}
{e(\tau)+\tilde{e}(\tau)}\tilde{e}(\tau),
\end{equation}
for all  $\tau'\in\operatorname{child}(\tau)$ and where we set \(\tilde{e}(\Ucal_r)\isdef e(\Ucal_r)\) for the root of the cluster tree. Given
this modified energy, we perform the thresholding version of the second 
algorithm from \cite{HM_BD97} with threshold 
\(\varepsilon^2\big\|{\bs f}^{\Sigma}_r\big\|^2\). This results in a subtree \(\Scal_{r}\) that approximates
\({\bs f}^\Delta_r\) up to a relative error 
of \(\varepsilon\) in the Euclidean
norm. Since the algorithm always selects either none or
all children of a given cluster, \(\Scal_{r}\) 
is a cluster tree and its leaves \(\Lcal(\Scal_{r})\)
form a partition of \(\Ucal_r\). Applying the preceding strategy to each subtree of the forest,
we obtain a signal reconstruction which approximates
\(f\) up to a relative error of \(\varepsilon\) with
respect to the Euclidean norm.

If the sustenance of the local tree structure is not
required, one may alternatively sort all coefficients
in the forest in decreasing order with respect to their
modulus and perform a relative norm-based thresholding, 
resulting in the, so called, best-$k$-term approximation
with respect to the given samplet basis.
\section{Graph partitioning and embedding} \label{sec:embedding}
This section serves as a bridge between the theoretical framework developed
so far and its practical realization. 

\subsection{Graph partitioning and coordinate maps}
To partition the graph into patches, we
employ multilevel heavy edge matching. Concretely, we rely on the $k$-way
partitioning implemented in \texttt{Metis}, see \cite{KK97,KK98}. We remark
that other approaches, such as recursive bisection using spectral clustering,
see, e.g., \cite{vLux07} would be possible as well. 

We assume that each of the resulting patches is an embedded manifold.
More specifically, the underlying manifold structure is embedded in a
certain {\em ambient space} $\mathbb{R}^d$, where possibly \(q\ll d\).
In general, the coordinate maps to embed a graph into the Euclidean space
$\mathbb{R}^q$, for a certain dimension \(q\), can then be obtained
by preserving the graph distances via multidimensional scaling,
see, e.g., \cite{book:BG05,Torgeson52} and the references therein. 
However, when an underlying manifold structure is present, we
can
exploit this and perform \textnormal{\texttt{Isomap}} \cite{TdSL00} which, roughly speaking, is a 
specific multidimensional scaling for graph embedded in manifolds.
For the resulting patches, we compare the moment matrix defined as in
\eqref{eq:mom_mat_new}, using the (unknown) embedding of the manifold,
and the moment matrix computed by the \textnormal{\texttt{Isomap}} approximation. 

\subsection{Embedding error bound}
The analysis of the embedding error is based on the results from \cite{ACJP20}.
Consider the points \(\bs x_1,\ldots,\bs x_{N_r} \in \Ucal_r \subset \Rbb^d\) stored in
the matrix \({\bs X}=[\bs x_1\cdots\bs x_{N_r}]^\intercal \in \Rbb^{{N_r}\times d}\). 
For the matrix \({\bs X}\), we define the \textit{half-width}
\(\omega({\bs X})\) as the smallest standard deviation along any direction in
space. As remarked in \cite{ACJP20}, this quantity is strictly positive if and only
if \(\bs X\) is of rank \(d\). In this case \(\omega({\bs X})
=\|\bs X^\dagger\|^{-1}/ \sqrt{{N_r}}\), with \(\bs X^\dagger\) being the
Moore-Penrose inverse. Finally, we define the
{\em maximum-radius} $\rho({\bs X})\isdef \max_{i \in \{1,\ldots,{N_r}\}}\| \bs x_i\|$.

Now, given \(\bs x_1,\ldots,\bs x_{N_r}\), \texttt{Isomap} computes points
\(\bs y_1,\ldots,\bs y_{N_r} \in \Rbb^q\) obtained by the diagonalization of the
double centered geodesic distance matrix ${\bs B} = -\frac{1}{2} {\bs H} {\bs D}
{\bs H}$, with ${\bs H} = {\bs I} - \frac{1}{{N_r}} {\bs 1}{\bs 1}^\intercal$ being the
centering matrix and ${\bs D}$ denoting the matrix of the squared geodesic
distances. 
Consequently, by considering the diagonalisation of the double centered geodesic
distance matrix ${\bs B} = {\bs U} {\bs \Lambda} {\bs U}^\intercal$, we obtain
${\bs y}_i = [(\sqrt{\lambda_1} {\bs u}_1)_i,\dots,(\sqrt{\lambda_q} {\bs u}_q)_i]^\intercal$ which are the low-dimensional representations of the initial points.

We remark that the geodesic distance \(d_{\Ucal_r}\), when the points are
quasi-uniformly distributed on the manifold, is well-approximated by the graph
distance, see \cite{BdSLT00}, whence we can use the latter instead of the real
(unknown)
geodesic distance.
In order to give a bound for the error of the embedding, we assume that the reach
\[
\begin{aligned}
\tau\isdef \sup \big\{t \geq 0 :\ \forall{\bs x} \in \mathbb{R}^d:\operatorname{dist}({\bs x}, \Ucal_r) &= t\\
\exists!{\bs v}\in \Ucal_r:\|{\bs x} -{\bs v}\| &= t \big\}
\end{aligned}
\]
of the manifold $\Ucal_r$ is fixed.
Moreover, we assume that
the points are sufficiently dense in the sense that there exists $a>0$ such that
$\min_i d_{\Ucal_r}({\bs x},{\bs x}_i)\leq a$ for each point $\bs x\in\Ucal_r$. 
In this setting, exploiting the result in \cite[Corollary 4]{ACJP20}, we obtain the following error bound.

\begin{proposition}\label{propV1}
    Let the reach of $\Ucal_r$ be fixed and let $t=c_1 a^{1/2}$, for a constant
    $c_1>0$. Further, let the points ${\bs x}_1,\ldots,{\bs x}_{N_r}$ be
    quasi-uniformly distributed on the manifold $\Ucal_r$.
    Denoting the embedded points obtained by \textnormal{\texttt{Isomap}} by ${\bs y}_1,\dots,{\bs y}_{N_r}$ 
    and the exact coordinates by ${\bs z}_1,\dots,{\bs z}_{N_r} \in \Rbb^q$,
    there holds
        \begin{align} \label{eq:errbound_isomap}
    \min_{Q \in \operatorname{O}(d)} 
\bigg[
\frac{1}{{N_r}} \!\!\sum_{i \in [{N_r}]} \| \bs y_i& - Q{\bs z}_i \|^2 \bigg]^{\frac{1}{2}} \!\!\!
\leq C\bigg( \frac{\log {N_r}}{{N_r}}\bigg)^{\frac 1 q}\!,
\end{align}
for a positive constant $C$.
\end{proposition}
\begin{proof}
The result follows from
\cite[Corollary 4]{ACJP20} under the made assumptions, cp.\
\cite[Section 5.1]{ACJP20}.
\end{proof}

\subsection{Consistency error bound for the moments}
When the graph is sampled from an underlying manifold, the coordinate
map produced by \texttt{Isomap} can be compared to the exact embedding provided by
the local coordinates whereof the manifold is endowed, as outlined in
\cite[Corollary 5]{ACJP20}. On the other hand, our analysis employs the samplet
transform on graphs, whose construction involves the definition of vanishing
moments on graphs, given in Section \ref{subsectionEF}. To verify the consistency
of our construction and justify the approximation of vanishing moments using
\texttt{Isomap}, we prove that the vanishing moments derived via \texttt{Isomap}
closely approximate those defined using the exact embedding.

\begin{proposition}
Under the assumptions of Proposition \ref{propV1}, for every multi-index $|\bs\alpha|\leq s$, there holds
 \begin{align}
    \min_{Q \in \operatorname{O}(d)} 
\bigg[
\frac{1}{{N_r}} \sum_{i \in [{N_r}]} \| \bs y_i^{\bs\alpha}- (Q \tilde{\bs z}_i)^{\bs \alpha} \|^2 \bigg]^{\frac{1}{2}}
\leq \Theta C\bigg( \frac{\log {N_r}}{{N_r}}\bigg)^{1/q},
\end{align}
    where $\Theta$ is the Lipschitz constant of the monomial $ \overline{\phi_r(\mathcal{U}_r)\cup\tilde{\phi}_r(\mathcal{U}_r)}\ni\bs x \mapsto \bs x^{\bs\alpha}\in\mathbb{R}$. Herein, \(\tilde{\phi}_r\) is
    the coordinate map obtained from \textnormal{\texttt{Isomap}}.
\end{proposition}
\begin{proof}
    The proof directly follows from the Lipschitz continuity the
    mapping of $\bs x\mapsto {\bs x}^{\bs \alpha}$ and the error bound \eqref{eq:errbound_isomap}.
\end{proof}

\begin{remark}
In our numerical examples, we replace \textnormal{\texttt{Isomap}} by its landmark version 
(\textnormal{\texttt{L-Iso\-map}}), see \cite{ST02}.
\textnormal{\texttt{L-Isomap}} is preferred over the classical \textnormal{\texttt{Iso\-map}},
because the latter is known to be computationally and memory-intensive for large
datasets. In turn, \textnormal{\texttt{L-Isomap}} reduces the computational and memory burden of
embedding large datasets by selecting \(n \ll N_r\) {\em landmark points}. The exact
version of \textnormal{\texttt{L-Isomap}} that we use is based on \cite{DST04}, where the landmarks are
chosen according to a \textnormal{\texttt{MaxMin}} greedy selection.
\end{remark}

\section{Numerical results}\label{sec:numerics}
 In the following, we benchmark the samplet compression on graphs 
 for different examples. 
For the \texttt{L-Isomap}, we select $100$ landmark
vertices per sub-graph, computed with the MinMax greedy algorithm, see \cite{DST04}. On each embedded patch,
we perform a geometric clustering by successively subdividing the bounding box of the embedded patch
along the longest axis. The resulting cluster tree
is then employed for the samplet construction.

For each example and each embedding dimension \(q\),
we report the following quantities, the 
\emph{lost energy}, the number of vanishing moments
\(s+1\) and the number of non-zero coefficients obtained via 
the algorithms discussed in Subsection~\ref{subsec:sparse}, i.e., 
\emph{nnz (AT)} for the adaptive tree and 
 \emph{nnz (NT)} for the relative norm-based thresholding.
To compute the lost energy, we first compute
the relative trace errors of the positive part of the
spectrum of the Gram matrix of the landmark points
for each patch and then take the maximum of this
value among patches.
In our context, relative norm-based thresholding consists of
setting the smallest coefficients of \(\bs{f}^\Sigma\) to zero
such that the fraction \((1-\varepsilon)\|\bs{f}^\Sigma\|\)
of the original norm is preserved.
In our experiments, we set the threshold for the relative
norm-based thresholding and for the adaptive tree both to 
\(\varepsilon = 10^{-2}\).
Hence, we obtain for both cases a reconstructed signal with
relative compression error, measured 
in the Euclidean norm, smaller than \(1\%\).

\subsection{Embedded unit square}
In the first example, we consider a nearest-neighbor
graph for points of the unit square \([0,1]\times[0,1]\)
embedded into \(\Rbb^{100}\). To this end, 
we consider \(N=10^6\) uniformly random points on the
unit square \([0,1]\times[0,1]\) which 
are embedded into the first two components 
of \(\Rbb^{100}\) and afterwards randomly rotated.
Afterwards, each data points is perturbed 
by additive uniform random noise in \(10^{-6}[-1,1]\).
The weighted graph is then obtained by connecting each
point to its neighbors within the
\(\varepsilon=4\cdot 10^{-3}\) ball around a given point.
The weights in the graph correspond to the respective
Euclidean distances between the nearest neighbors.
This results in an averaged valence of \(50\) and
a graph with a total number of edges equal to
\(50\,110\,170\). As signal, we consider the function
\(f(\bs x)=\mathrm{exp}(-4 \|\bs x-\bs x^{*}\|)\,\cos(8 \pi \|\bs x-\bs x^{*}\|)\), where $\bs x^*$ is the center of the embedded and rotated point \([0.5,0.5]\).
The signal on the embedded graph is shown in
Figure~\ref{fig:M_square}. In this example, the samplet
forest contains a single samplet tree, i.e., we represent
the graph by a single patch. Especially, no graph
partitioning is performed here.

In Table~\ref{tab:results_M_square} we report the results
obtained via the samplet compression of the signal.
As can be inferred from the table, by increasing the embedding dimension \(q=2,3,4\), we observe a decrease
in the lost energy, since more eigenvalues of the
Gram matrix are captured as a consequence of the high-dimensional embedding. In all cases, the lost energy
is around \(10^{-3}\), which is expected, as we consider
a flat two-dimensional manifold perturbed by noise.
Especially, increasing the embedding dimension does not
significantly decrease the energy loss.

We observe a significant decrease
in both cases of nnz with increasing number of 
vanishing moments.
Regarding the number of nnz (AT), for \(q=2\), samplets with one vanishing moment, corresponding to Haar wavelets on trees, see \cite{gavish}, require \(260\,302\) 
non-zero coefficients to reach the threshold, while samplets with 5 vanishing moments only require \(1\,081\)
coefficients. The resulting compression is better by a factor of roughly \(240\). For \(q=3,4\), it becomes slightly worse, however, increasing the number of vanishing moments still yields much sparser representations. 
With norm-based thresholding, we observe the same trend as for nnz (AT). However, the sparsification of the reconstructed signal improves by a factor of \(2\) compared to the adaptive tree algorithm when \(q=2\), and by approximately \(2.5\) when \(q=3,4\).
\begin{remark}
It is notable, in this example, the critical role of working in the parametric space to mitigate the curse of dimensionality. For instance, if we consider the maximum number of vanishing moments, i.e., \( s+1 = 5 \), the cardinality \( m_s \) of the polynomial space in the ambient space is \(m_s = 4\,598\,126\) while in the parametric space it reduces to only \(70\).
This substantial reduction significantly alleviates the number of conditions that must be satisfied to achieve the same number of vanishing moments, by restricting the construction to the parametric dimension of the manifold rather than the ambient dimension.
\end{remark}

\begin{figure}[h]
\centering
\begin{tikzpicture}
  \draw(0.5,0) node{\includegraphics[scale = 0.05,clip,trim = 500 500 300 500]{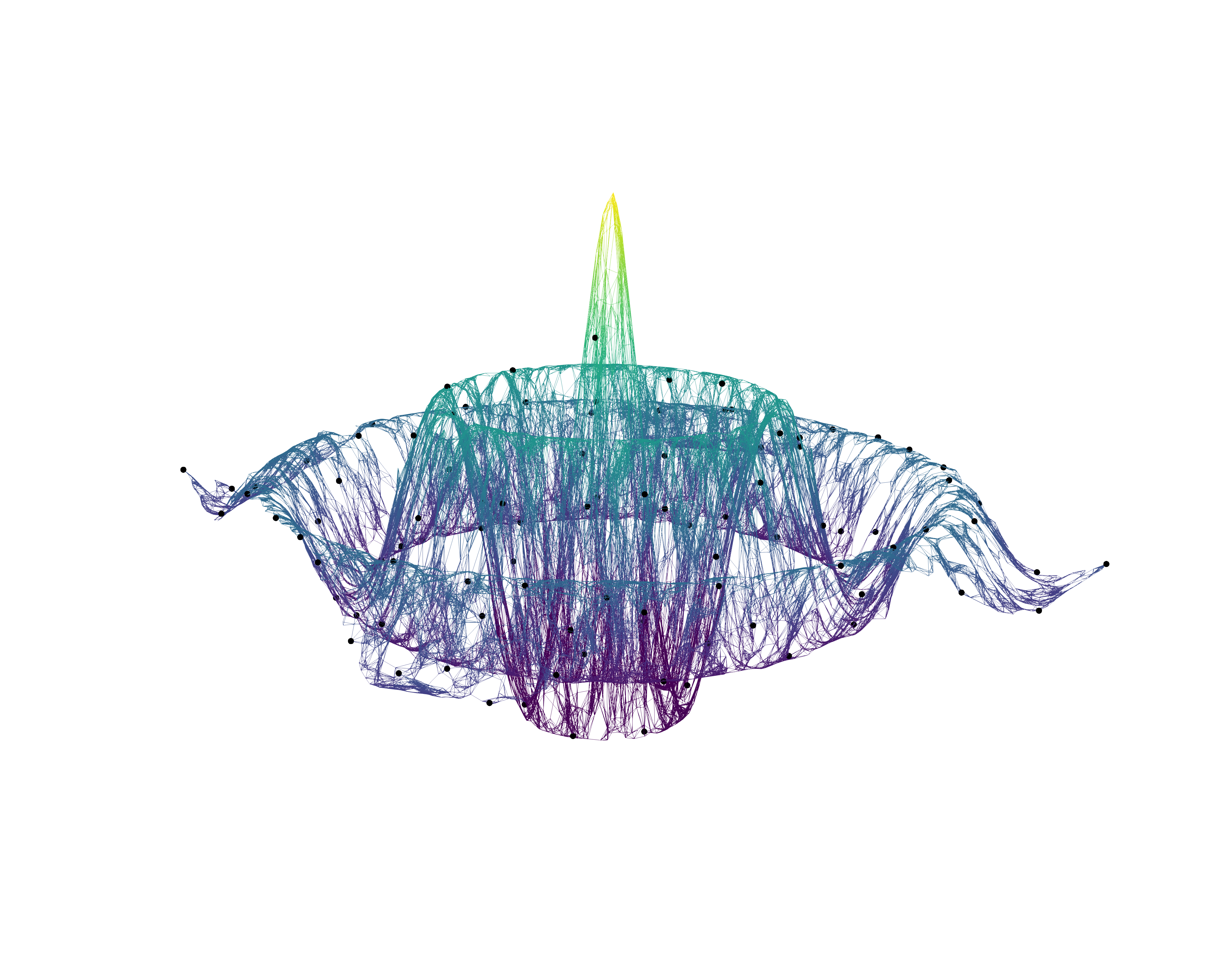}};
\end{tikzpicture}
\caption{Signal on the unit square. Black dots correspond to landmark
  vertices of the graph used for \texttt{L-Isomap}.}\label{fig:M_square}
\end{figure}

\begin{table}[h]
  \centering
  \caption{Results for the embedded unit square.}
  \label{tab:results_M_square}
    \begin{tabular}{ccrr}
      \toprule
       lost energy & $s+1$ & nnz (AT) & nnz (NT) \\
      \midrule
      \multicolumn{4}{l}{\(q=2\)} \\
      \midrule
      $1.27\cdot10^{-3}$ & 1 & 260\,302 & 153\,090 \\
      $1.27\cdot10^{-3}$ & 2 &  8\,620  & 4\,836 \\
      $1.27\cdot10^{-3}$ & 3 &  2\,535  & 1\,371 \\
      $1.27\cdot10^{-3}$ & 4 &  1\,423  & 796 \\
      $1.27\cdot10^{-3}$ & 5 &  1\,081  & 588 \\
      \midrule
      \multicolumn{4}{l}{\(q=3\)} \\
      \midrule
      $1.22\cdot10^{-3}$ & 1 & 340\,214 & 164\,204 \\
      $1.22\cdot10^{-3}$ & 2 & 17\,982  & 5\,491 \\
      $1.22\cdot10^{-3}$ & 3 &  4\,035  & 1\,674 \\
      $1.22\cdot10^{-3}$ & 4 &  2\,792  & 1\,236 \\
      $1.22\cdot10^{-3}$ & 5 &  2\,479  & 1\,034 \\
      \midrule
      \multicolumn{4}{l}{\(q=4\)} \\
      \midrule
      $1.13\cdot10^{-3}$ & 1 & 434\,124 & 182\,164 \\
      $1.13\cdot10^{-3}$ & 2 & 36\,372  & 8\,420 \\
      $1.13\cdot10^{-3}$ & 3 &  6\,665  & 2\,154 \\
      $1.13\cdot10^{-3}$ & 4 &  4\,695  & 1\,894 \\
      $1.13\cdot10^{-3}$ & 5 &  4\,832  & 1\,877 \\
      \bottomrule
    \end{tabular}%
\end{table}

\subsection{Swiss roll}
In the second example, we consider the Swiss roll
\begin{equation}
[x(t,v),y(t,v),z(t,v)]=[t\cos t,v,t\sin t],
\end{equation}
where we randomly sample \(N=10^6\)
pairs $(t,v) \in [1.5\pi,4.5\pi]\times[0,10]$ with
respect to the uniform distribution. Moreover, we 
center the resulting data sites with respect to their
mean value and rescale them such that the longest edge
of their bounding box has length 1. The resulting bounding
box is \([-9.48,12.61]\times[0,10]\times[-11.04,14.14]\).
As before, a graph is obtained by connecting each point
to its \(\varepsilon = 5\times10^{-3}\) nearest neighbors. 
This results in a total of \(52\,951\,614\) edges and an average 
valence of \(52\). 
The
weights in the graph correspond to the respective
Euclidean distances between the nearest neighbors.
The signal is given by 
\(f(\bs x)=\cos(\pi \|\bs x-\bs x^{*}\|)\),
where \(\mathbf{x}^* = [-3\pi,\,5,\,0]\) lies on the manifold. 
Figure~\ref{fig:SwissRoll} displays the resulting signal and the
chosen landmark points. For graph partitioning we employ again a single patch. 

The results are shown in Table~\ref{tab:results_SR}. 
As the embedding dimension increases from \(2\) to \(4\),
the lost energy only slightly decreases from \(8.22\times10^{-4}\)
to \(6.40\times10^{-4}\). The number of non-zeros coefficients,
exhibits a non‐monotonic dependence on \(s+1\). It is the largest for
one vanishing moment, drops sharply to a minimum at \(s+1=3\) for all
embedding dimensions. Then increases again for \(s+1=4,5\). This
indicates that samplets with three vanishing moments strike the best
balance between locality and expressiveness for the signal at hand.
Norm-based thresholding exhibits the same trend as the adaptive tree algorithm but achieves better performance in terms of sparse representation by a factor approximately of \(2\). This improvement in nnz (NT) can be attributed to the pointwise nature of norm-based thresholding, which allows for selectively retaining or discarding individual coefficients, whereas the adaptive tree strategy collects all coefficients within a subtree, including those with small magnitudes, if they belong to a quasi-optimal subtree for the sparse representation of the signal.

\begin{figure}[h]
\centering
\begin{tikzpicture}
  \draw(-0,0) node{\includegraphics[scale = 0.055,clip,trim = 800 300 800 300]{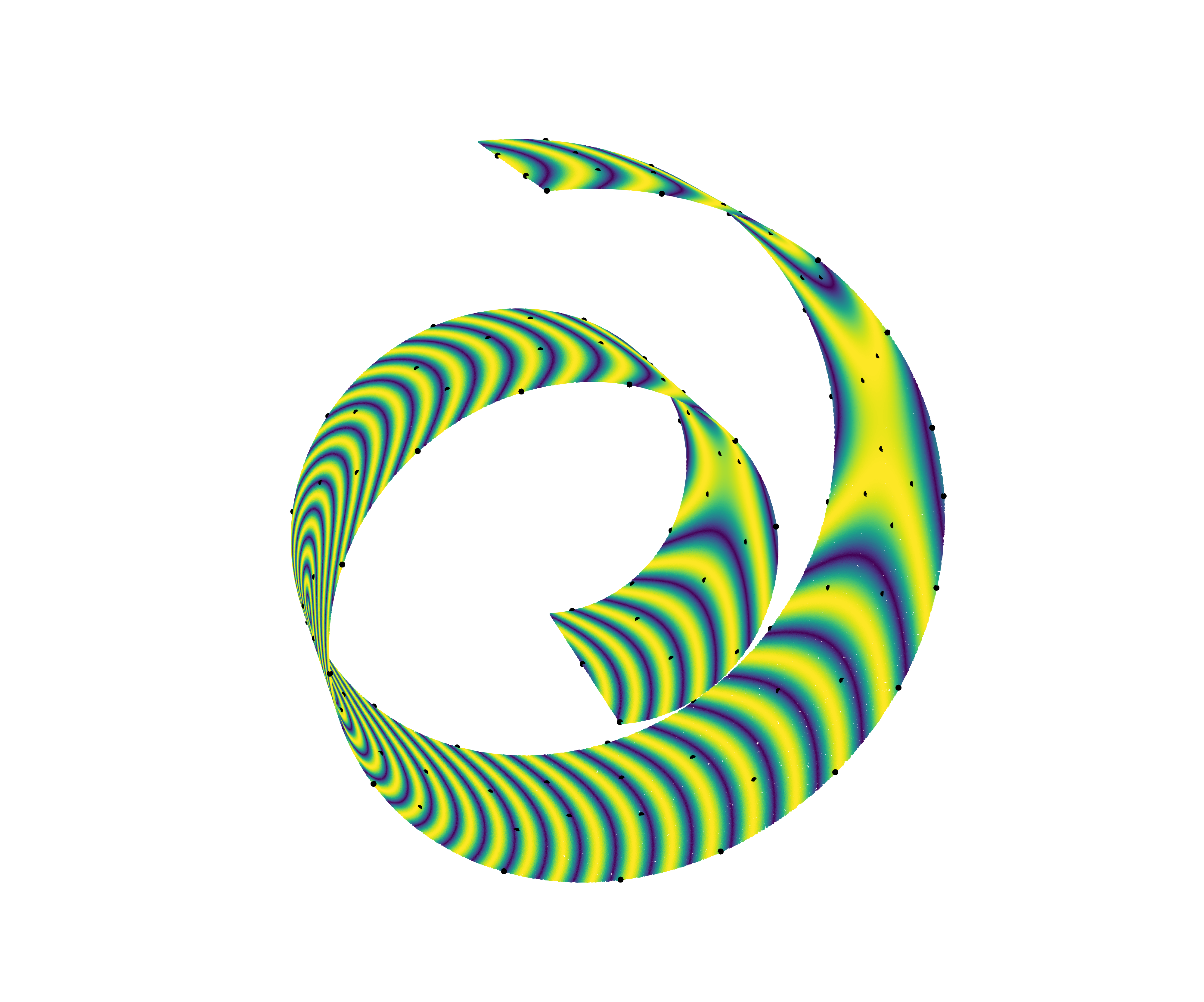}};
\end{tikzpicture}
\caption{Signal on the Swiss roll. Black dots correspond to landmark vertices
  of the graph used for \texttt{L-Isomap}.}\label{fig:SwissRoll}
\end{figure}

\begin{table}[h]
  \centering
  
  \caption{Results for the Swiss roll.}
  \label{tab:results_SR}
    \begin{tabular}{ccrr}
      \toprule
       lost energy & $s+1$ & nnz (AT) & nnz (NT) \\
      \midrule
      \multicolumn{4}{l}{$q = 2$} \\
      \midrule
      $8.22\cdot10^{-4}$ & 1 & 507\,059 & 301\,797 \\
      $8.22\cdot10^{-4}$ & 2 &  82\,979 & 42\,318 \\
      $8.22\cdot10^{-4}$ & 3 &  49\,098 & 23\,108 \\
      $8.22\cdot10^{-4}$ & 4 &  53\,141 & 22\,464 \\
      $8.22\cdot10^{-4}$ & 5 &  58\,974 & 23\,152 \\
      \midrule
      \multicolumn{4}{l}{$q = 3$} \\
      \midrule
      $7.11\cdot10^{-4}$ & 1 & 598\,974 & 293\,488 \\
      $7.11\cdot10^{-4}$ & 2 & 144\,267 & 60\,288 \\
      $7.11\cdot10^{-4}$ & 3 &  86\,631 & 39\,487 \\
      $7.11\cdot10^{-4}$ & 4 & 120\,490 & 46\,090 \\
      $7.11\cdot10^{-4}$ & 5 & 133\,278 & 54\,495 \\
      \midrule
      \multicolumn{4}{l}{$q = 4$} \\
      \midrule
      $6.40\cdot10^{-4}$ & 1 & 670\,964 & 281\,488 \\
      $6.40\cdot10^{-4}$ & 2 & 216\,569 & 80\,201 \\
      $6.40\cdot10^{-4}$ & 3 & 146\,767 & 62\,766 \\
      $6.40\cdot10^{-4}$ & 4 & 186\,112 & 83\,150 \\
      $6.40\cdot10^{-4}$ & 5 & 269\,142 & 112\,265 \\
      \bottomrule
    \end{tabular}%
\end{table}

\subsection{Stanford bunny}
As the third example, we consider a point cloud of the Stanford bunny,
containing  \(N=911\,990\) points located at the bunny's surface.
The bounding box is given by \([-0.44,0.56]\times[-0.39,0.59]\times[-0.45,0.32]\).
As in the previous examples, a graph is obtained by connecting each
point to its \(\varepsilon=6\times10^{-3}\) nearest neighbors. 
As in the previous cases, the weights in the graph correspond
to the respective Euclidean distances between the nearest neighbors.
The resulting graph has an average valence of \(44\) with a total number
of edges equal to \(40\,390\,966\).
We consider the signal \(f(\bs x)
=\mathrm{exp}(-10 \|\bs x-\bs x^{*}\|)\,\cos(20 \pi \|\bs x-\bs x^{*}\|)\),
where \(\mathbf{ \bs x}^*\) has been selected randomly and around the neck,
as shown in Figure~\ref{fig:S_bunny}. For this more intricate manifold,
we employ different values for the number of patches, i.e., \(p=50,100,150,200\). 
In Table~\ref{tab:results_bunny_d2}, we report the results for embedding
dimension \(q=2\). The lost energy significantly decreases from
\(0.3607\) for \(p=50\) to \(0.1037\) for \(p=100\), reflecting a reduced
distortion with smaller patches. The number increases again for more patches.
Increasing the number of vanishing moments shows consistently an improved
compression, with the best compression being achieved for \(200\) patches
and \(s+1=5\). Table~\ref{tab:results_bunny_d3} reports results for 
embedding dimension \(q=3\), while  Table~\ref{tab:results_bunny_d4} reports
them for \(q=4\). Qualitatively, they are similar to the \(q=2\) case, 
although the lost energy is smaller due to the higher embedding dimension.
The overall best compression is achieved for embedding dimension \(q=3\)
with \(100\) patches and \(5\) vanishing moments. 
As expected, we observe similar results for nnz (NT) in this final example where the overall best compression is achieved with the same number of patches and embedding dimension, and the general trend remains consistent. Notably, the improvement factor of nnz (NT) over nnz (AT) is larger than in the previous examples, resulting in a better sparsification of the signal.

\begin{figure}
\centering
\begin{tikzpicture}
  \draw(0.5,0) node{\includegraphics[scale = 0.035,clip,trim = 700 200 750 350]{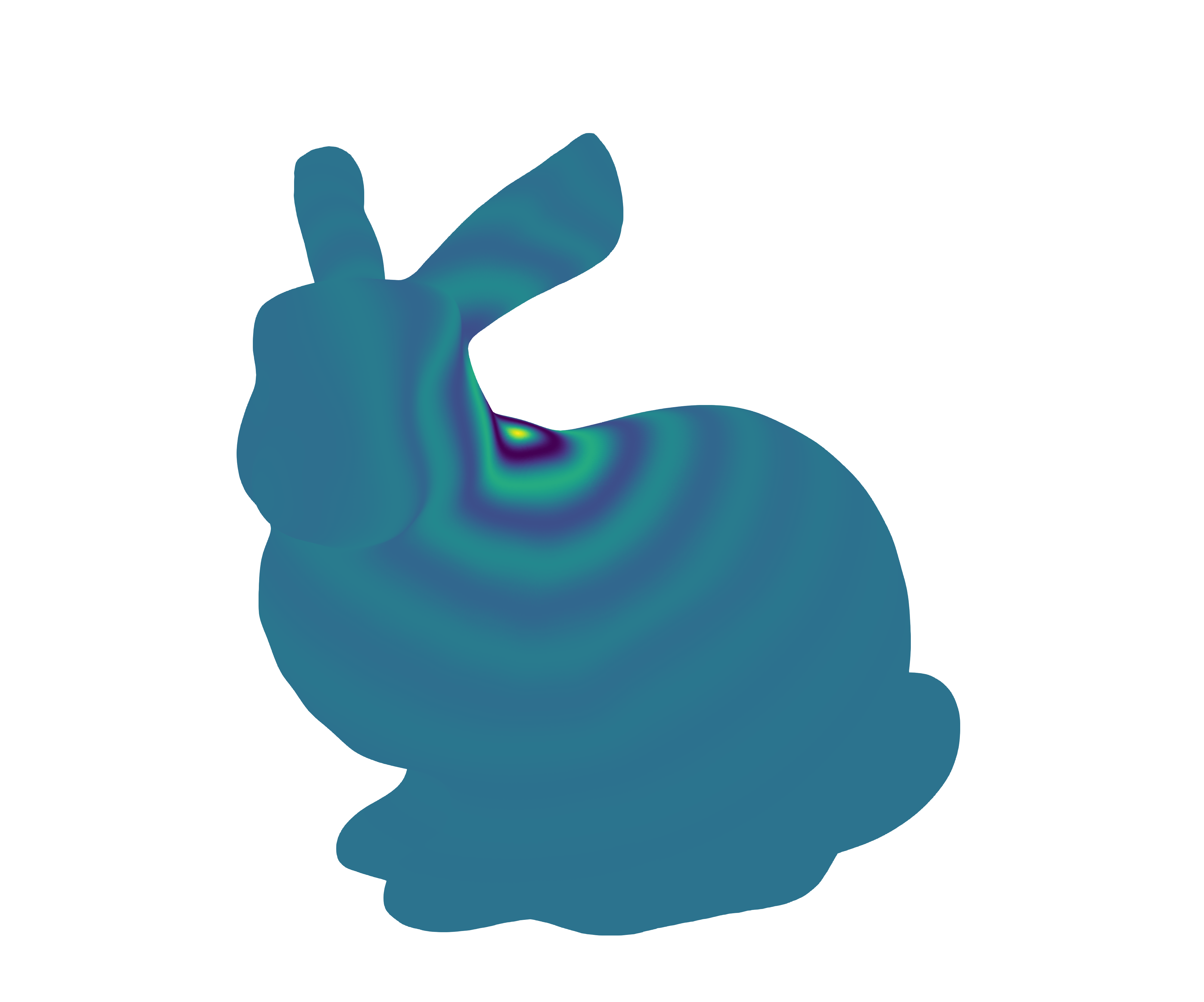}};
  \draw(-1.5,-3) node{\includegraphics[scale = 0.02,trim = 700 200 750 350, clip]{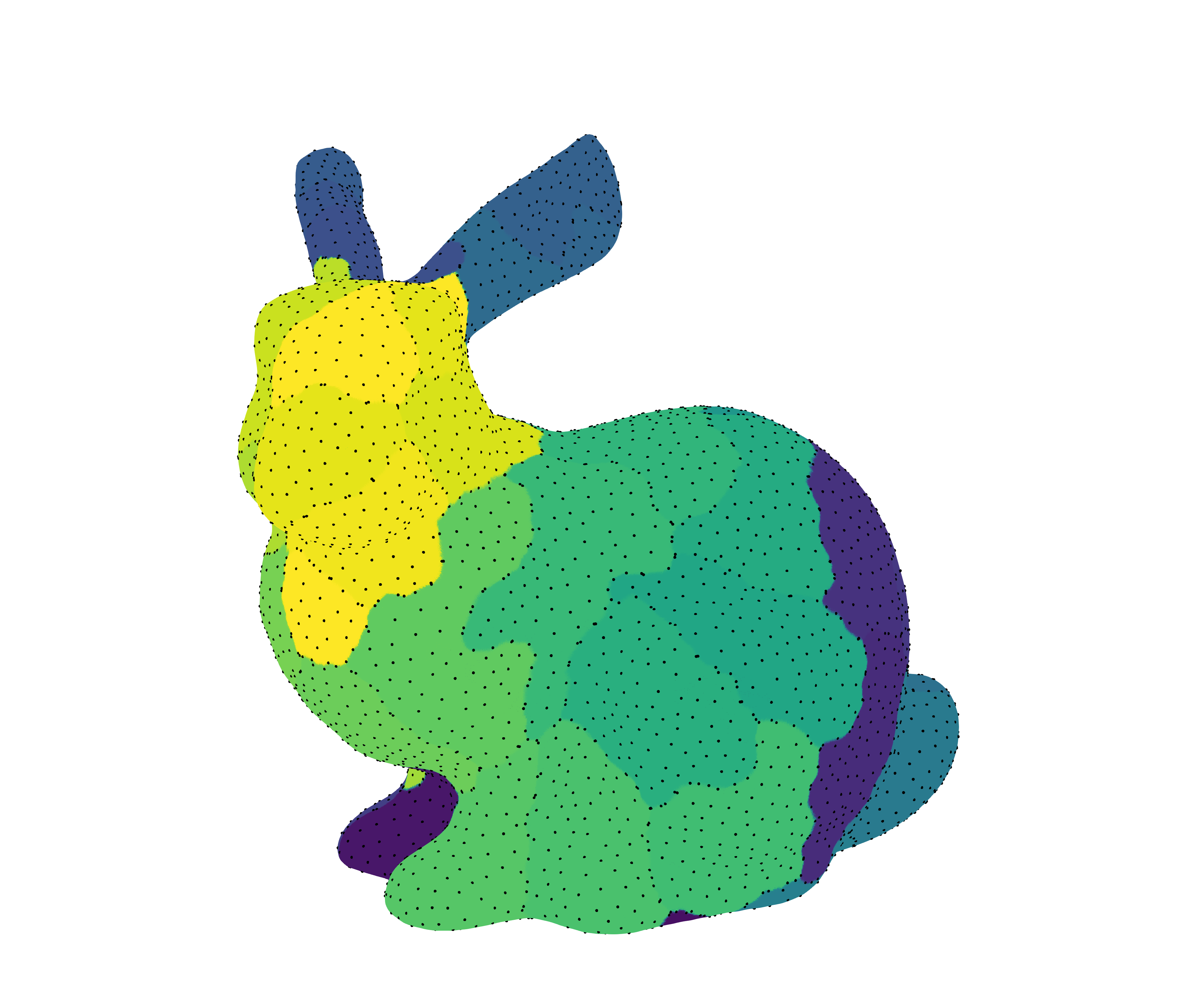}};
    \draw(2.5,-3) node{\includegraphics[scale = 0.02,trim = 700 200 750 350, clip]{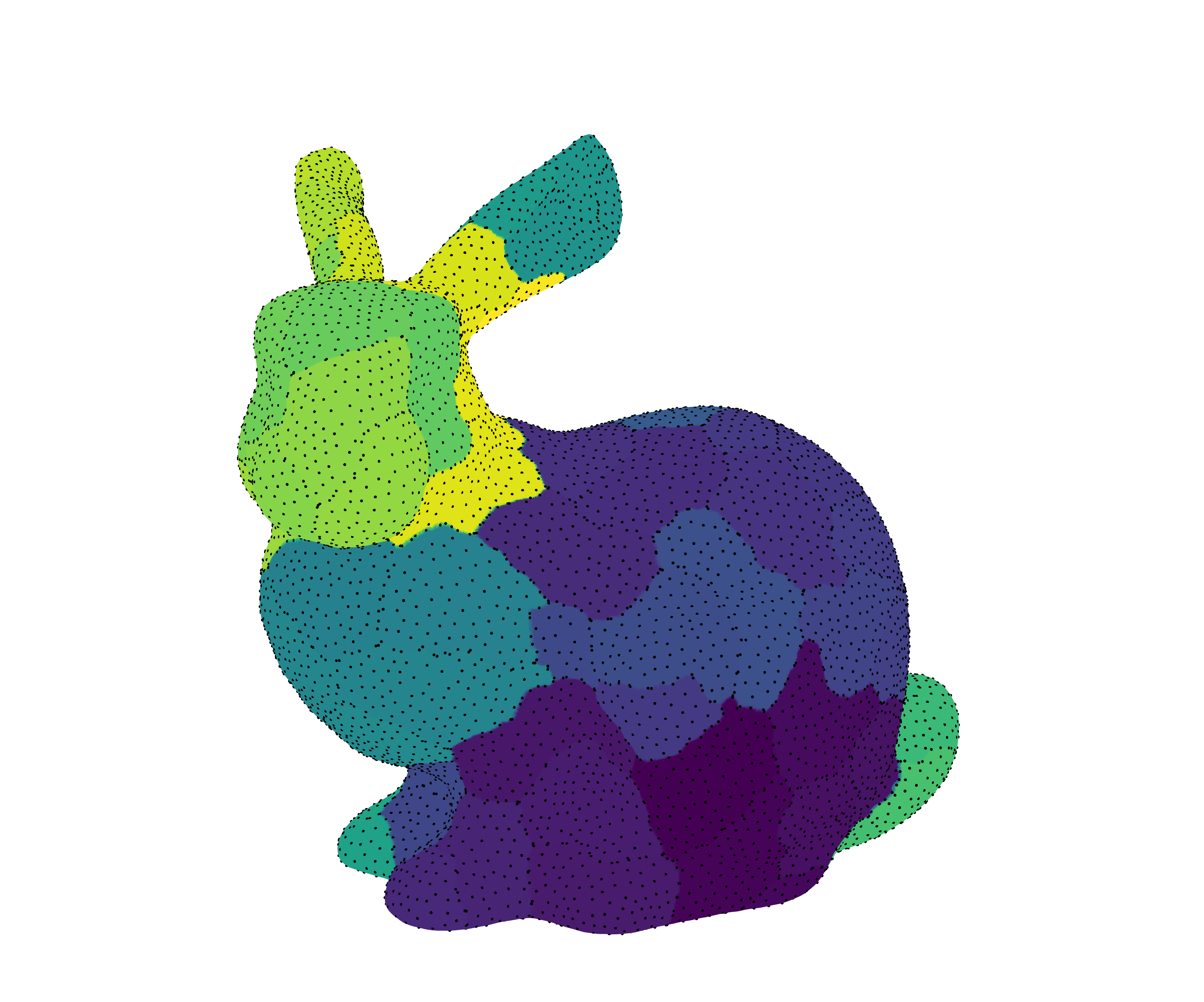}};
    \draw(-1.5,-5.5) node{\includegraphics[scale = 0.02,clip,trim = 700 200 750 350]{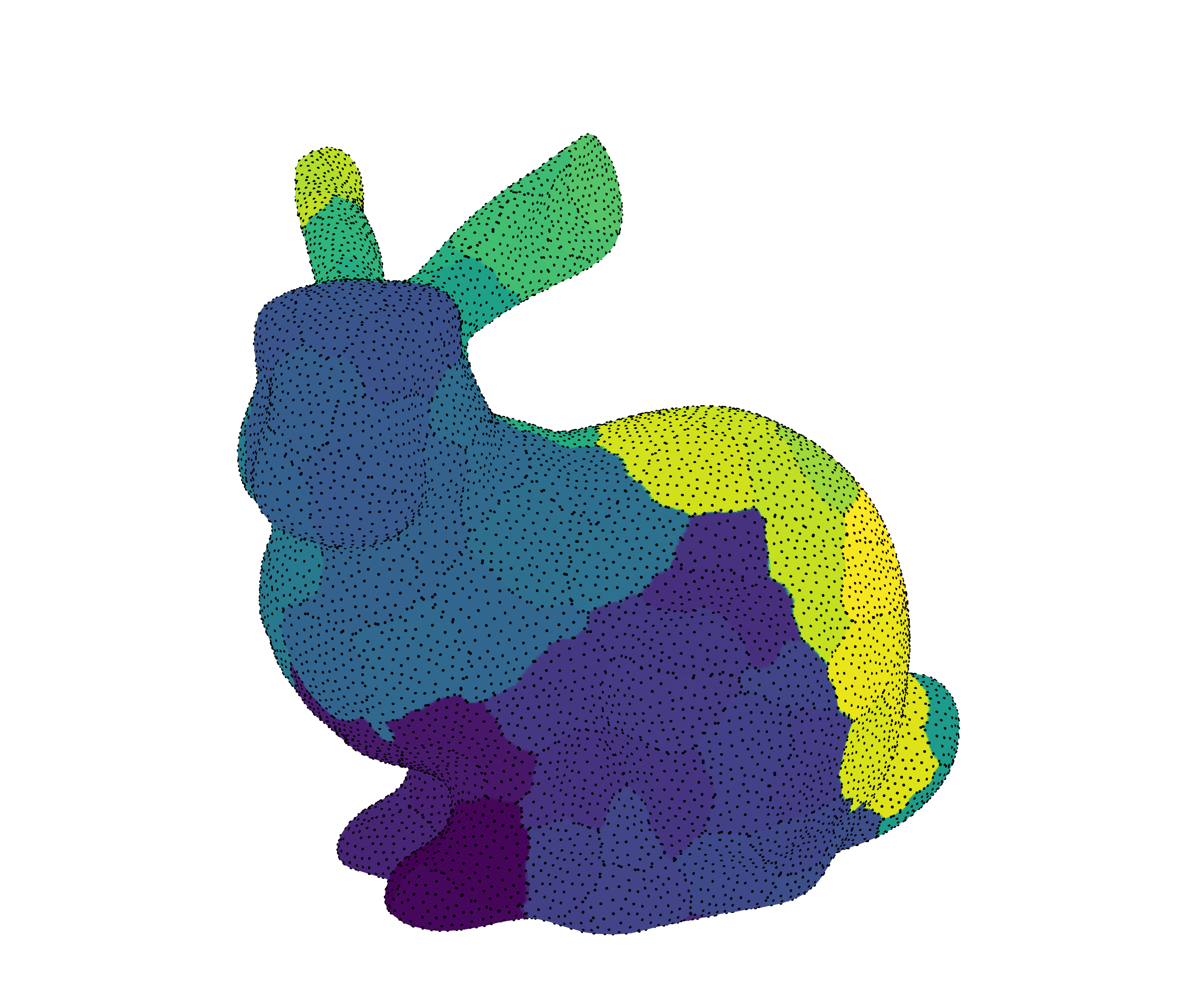}};
    \draw(2.5,-5.5) node{\includegraphics[scale = 0.02,clip,trim = 700 200 750 350]{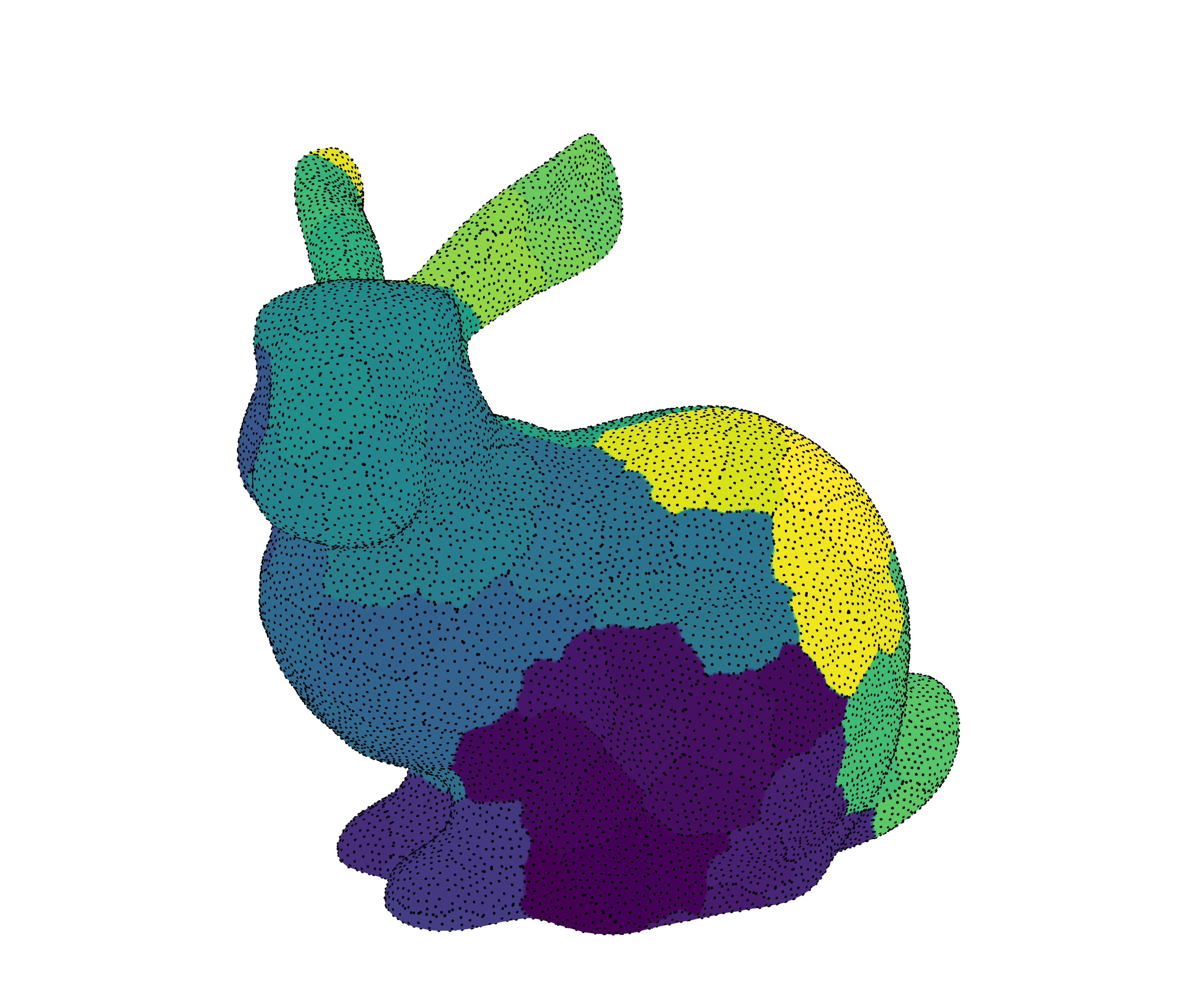}};
\end{tikzpicture}
\caption{Top image shows the signal on the Stanford bunny.
The second row shows the partitions for 50 and 100 patches, 
while the bottom row shows them for 150 and 200 patches. The
black dots correspond to landmarks.
} \label{fig:S_bunny}
\end{figure}

\begin{table}[htb]
  \centering
  \caption{Stanford bunny performances for embedding dimension 2.}
  \label{tab:results_bunny_d2}
    \begin{tabular}{ccrr}
      \toprule
       lost energy & $s+1$ & nnz (AT) & nnz (NT) \\
    \midrule
      \multicolumn{4}{l}{$p = 50$} \\
      \midrule
      $3.64\cdot10^{-1}$  & 1 & 593\,983 & 175\,724  \\
      $3.64\cdot10^{-1}$  & 2 & 184\,098 & 70\,578  \\
      $3.64\cdot10^{-1}$  & 3 & 141\,327 & 54\,564  \\
      $3.64\cdot10^{-1}$  & 4 & 132\,760 & 50\,505  \\
      $3.64\cdot10^{-1}$  & 5 & 131\,225 & 49\,352  \\
      \midrule
      \multicolumn{4}{l}{$p = 100$} \\
      \midrule
      $1.04\cdot10^{-1}$  & 1 & 614\,539 & 161\,938  \\
      $1.04\cdot10^{-1}$  & 2 & 139\,151 & 44\,214  \\
      $1.04\cdot10^{-1}$  & 3 &  89\,202 & 26\,102  \\
      $1.04\cdot10^{-1}$  & 4 &  80\,204 & 22\,021  \\
      $1.04\cdot10^{-1}$  & 5 &  76\,573 & 20\,235  \\
      \midrule
      \multicolumn{4}{l}{$p = 150$} \\
      \midrule
      $1.96\cdot10^{-1}$  & 1 & 614\,274 & 139\,510  \\
      $1.96\cdot10^{-1}$  & 2 & 109\,810 & 30\,029  \\
      $1.96\cdot10^{-1}$  & 3 &  57\,044 & 17\,448  \\
      $1.96\cdot10^{-1}$  & 4 &  47\,741 & 14\,559  \\
      $1.96\cdot10^{-1}$  & 5 &  44\,175 & 13\,604  \\
      \midrule
      \multicolumn{4}{l}{$p = 200$} \\
      \midrule
      $1.72\cdot10^{-1}$  & 1 & 614\,563 & 137\,580  \\
      $1.72\cdot10^{-1}$  & 2 &  95\,031 & 22\,604  \\
      $1.72\cdot10^{-1}$  & 3 &  38\,977 & 9\,285  \\
      $1.72\cdot10^{-1}$  & 4 &  28\,997 & 6\,600  \\
      $1.72\cdot10^{-1}$  & 5 &  25\,960 & 5\,608  \\
      \bottomrule
    \end{tabular}%
\end{table}

\begin{table}[!t]
  \centering
  
  \caption{Stanford bunny performances for embedding dimension 3.}
  \label{tab:results_bunny_d3}
    \begin{tabular}{ccrr}
      \toprule
       lost energy & $s+1$ & nnz (AT) & nnz (NT) \\
    \midrule
      \multicolumn{4}{l}{$p = 50$} \\
      \midrule
      $1.63\cdot10^{-1}$  & 1 & 579\,753 & 128\,082 \\
      $1.63\cdot10^{-1}$  & 2 & 123\,468 & 30\,416 \\
      $1.63\cdot10^{-1}$  & 3 &  56\,157 & 14\,018 \\
      $1.63\cdot10^{-1}$  & 4 &  39\,644 & 8\,106 \\
      $1.63\cdot10^{-1}$  & 5 &  31\,686 & 5\,799 \\
      \midrule
      \multicolumn{4}{l}{$p = 100$} \\
      \midrule
      $6.78\cdot10^{-1}$  & 1 & 621\,074 & 124\,617 \\
      $6.78\cdot10^{-1}$  & 2 & 122\,618 & 24\,005 \\
      $6.78\cdot10^{-1}$  & 3 &  46\,430 & 7\,438 \\
      $6.78\cdot10^{-1}$  & 4 &  27\,864 & 4\,941 \\
      $6.78\cdot10^{-1}$  & 5 &  22\,591 & 4\,337 \\
      \midrule
      \multicolumn{4}{l}{$p = 150$} \\
      \midrule
      $1.23\cdot10^{-1}$  & 1 & 623\,314 & 122\,802 \\
      $1.23\cdot10^{-1}$   & 2 & 118\,297 & 19\,378 \\
      $1.23\cdot10^{-1}$   & 3 &  41\,032 & 7\,267 \\
      $1.23\cdot10^{-1}$   & 4 &  26\,835 & 5\,226 \\
      $1.23\cdot10^{-1}$   & 5 &  25\,300 & 4\,725 \\
      \midrule
      \multicolumn{4}{l}{$p = 200$} \\
      \midrule
      $1.13\cdot10^{-1}$  & 1 & 637\,392 & 116\,987 \\
      $1.13\cdot10^{-1}$  & 2 & 119\,877 & 19\,421 \\
      $1.13\cdot10^{-1}$  & 3 &  41\,577 & 7\,090 \\
      $1.13\cdot10^{-1}$  & 4 &  27\,780 & 5\,167 \\
      $1.13\cdot10^{-1}$  & 5 &  25\,098 & 4\,633 \\
      \bottomrule
    \end{tabular}%
\end{table}

\begin{table}[!t]
  \centering
  
  \caption{Stanford bunny performances for embedding dimension 4.}
  \label{tab:results_bunny_d4}
    \begin{tabular}{ccrr}
      \toprule
       lost energy & $s+1$ & nnz (AT) & nnz (NT) \\
    \midrule
      \multicolumn{4}{l}{$p = 50$} \\
      \midrule
      $1.22\cdot10^{-1}$  & 1 & 600\,164 & 116\,623 \\
      $1.22\cdot10^{-1}$  & 2 & 145\,580 & 22\,163 \\
      $1.22\cdot10^{-1}$  & 3 &  46\,728 & 7\,794 \\
      $1.22\cdot10^{-1}$  & 4 &  27\,178 & 5\,743 \\
      $1.22\cdot10^{-1}$  & 5 &  23\,660 & 5\,418 \\
      \midrule
      \multicolumn{4}{l}{$p = 100$} \\
      \midrule
      $4.90\cdot10^{-2}$  & 1 & 634\,387 & 114\,572 \\
      $4.90\cdot10^{-2}$  & 2 & 140\,315 & 19\,567 \\
      $4.90\cdot10^{-2}$  & 3 &  45\,630 & 7\,459 \\
      $4.90\cdot10^{-2}$  & 4 &  29\,001 & 5\,590 \\
      $4.90\cdot10^{-2}$  & 5 &  26\,772 & 5\,256 \\
      \midrule
      \multicolumn{4}{l}{$p = 150$} \\
      \midrule
      $9.23\cdot10^{-2}$  & 1 & 650\,262 & 106\,155 \\
      $9.23\cdot10^{-2}$  & 2 & 144\,264 & 18\,023 \\
      $9.23\cdot10^{-2}$  & 3 &  45\,693 & 7\,344 \\
      $9.23\cdot10^{-2}$  & 4 &  31\,406 & 5\,857 \\
      $9.23\cdot10^{-2}$  & 5 &  30\,037 & 5\,636 \\
      \midrule
      \multicolumn{4}{l}{$p = 200$} \\
      \midrule
      $8.19\cdot10^{-2}$  & 1 & 651\,375 & 108\,706 \\
      $8.19\cdot10^{-2}$  & 2 & 145\,599 & 18\,747 \\
      $8.19\cdot10^{-2}$  & 3 &  49\,647 & 7\,788 \\
      $8.19\cdot10^{-2}$  & 4 &  35\,239 & 6\,102 \\
      $8.19\cdot10^{-2}$  & 5 &  34\,132 & 5\,860 \\
      \bottomrule
    \end{tabular}%
\end{table}

\section{Conclusions}\label{sec:conclusion}
We have extended the samplet framework to perform multiresolution analysis on 
graph signals. 
The resulting method provides locality, orthogonality, and higher-order vanishing 
moments. This makes it particularly effective for large-scale applications
involving graphs that exhibit an underlying manifold structure,
offering a practical tool for signal analysis across various domains.
Particularly, we have shown that any graph signal that can locally be
approximated by generalized polynomials admits a rapidly decaying samplet
expansion. In our numerical examples, the suggested approach consistently delivers
compression ratios orders of magnitude higher than classical Haar‐wavelets in both synthetic manifold examples and general graph settings.

\section*{Acknowledgment}
The authors have been funded by the Swiss National Science
Foundation starting grant
``Multiresolution methods for unstructured data'' (TMSGI2\_211684).

\bibliographystyle{plain}
\bibliography{literature}

\end{document}